\documentclass[a4paper,english]{lipics-v2018}

\nolinenumbers

\newcommand{\minus}{\vspace{-1mm}}
\newcommand{\ignore}[1]{}
\newcommand{\para}[1]{\subparagraph{#1}}  %\bigskip\noindent{\bf #1}\xspace}
\newcommand{\whole}{{}}

\newcommand{\clo}{\mathrm{c}}

\newcommand{\str}{{\cal A}}
\newcommand{\unitint}{\mathbb{I}}

\newcommand{\A}{{\cal A}^\whole}
\newcommand{\Aclock}{{\cal A}^\clo}

\newcommand{\AZ}{{\cal A}^\whole_\Z}
\newcommand{\AZc}{{\cal A}^\clo_\N}

\newcommand{\AQ}{{\cal A}^\whole_\Q}

\newcommand{\AIc}{{\cal A}^\clo_\I}
\newcommand{\logic}{\ensuremath{\LL^\whole_{\Z, \Q}}\xspace}

\newcommand{\logicint}{\LL^\whole_{\Z}}
\newcommand{\logicrat}{\LL^\whole_{\Q}}

\newcommand{\logicclock}{\LL^\clo_{\N, \unitint}}
\newcommand{\logicclockint}{\LL^\clo_{\N}}
\newcommand{\logicclockrat}{\LL^\clo_{\unitint}}

\renewcommand{\iff}{\Leftrightarrow}
\newcommand{\cycminus}{\ominus}
\newcommand{\cycplus}{\oplus}
\newcommand{\condone}[1]{\mathbbm 1_{#1?}}

\newcommand{\Qgeq}{\Q_{\geq 0}}
\newcommand{\goesto}[1]{\xrightarrow{#1}}
\newcommand{\TPDArule}[3]{\tuplesmall {#1 , #2, #3}}
\newcommand{\lang}[1]{L(#1)}
\newcommand{\languntimed}[1]{L^{\textrm{un}}(#1)}

\newcommand{\PDA}{\mbox{\sc pda}\xspace}
\newcommand{\TPDA}{\mbox{\sc tpda}\xspace}
\newcommand{\RPDA}{\mbox{\sc rpda}\xspace}
\newcommand{\PDTA}{\mbox{\sc ptda}\xspace}
\newcommand{\dtPDA}{dt\PDA}
\newcommand{\TA}{\mbox{\sc ta}\xspace}
\newcommand{\RTA}{\mbox{\sc rta}\xspace}
\newcommand{\TRPDA}{\mbox{\sc trpda}\xspace}
\newcommand{\true}{\mathbf{true}}
\newcommand{\false}{\mathbf{false}}
\newcommand{\N}{\mathbb N}
\newcommand{\Z}{\mathbb Z}
\newcommand{\Q}{\mathbb Q}

\newcommand{\I}{\mathbb I}
\newcommand{\C}{\mathcal C}

\newcommand{\LL}{\mathcal L}
\newcommand{\op}{\mathsf{op}}

\newcommand{\elapse}{\mathsf{elapse}}
\newcommand{\pushop}[1]{\mathsf{push}(#1)}
\newcommand{\popop}[1]{\mathsf{pop}(#1)}

\newcommand{\pushopgen}[2]{{\mathsf{push}(#1 : #2)}}
\newcommand{\popopgen}[2]{{\mathsf{pop}(#1 : #2)}}
\newcommand{\fract}[1]{\left\{#1\right\}}
\renewcommand{\P}{\mathcal P}
\newcommand{\QQ}{\mathcal Q}
\newcommand{\restrict}[2]{\left.#1\right|_{#2}}
\renewcommand{\vec}[1]{\overline{#1}}

\newcommand{\sem}[1]{\left\ldbrack#1\right\rdbrack}

\newcommand{\orbits}[2]{\mathrm{Orb}^{\ignore{#1}}(#2)}
\newcommand{\floor}[1]{\lfloor#1\rfloor}

\newcommand{\tuple}[1]{\left\langle#1\right\rangle}
\newcommand{\tuplesmall}[1]{\langle#1\rangle}
\newcommand{\eqv}[1]{\equiv_{#1}}
\newcommand{\eqvs}{(\eqv m)_{m\in\N}}

\newcommand{\reach}[2]{\stackrel{#1}{\leadsto}_{#2}} %{\leadsto_{#1}}%
\newcommand{\reachp}[2]{\stackrel{#1}{\leadsto}'_{#2}} %{\leadsto_{#1}}%
\newcommand{\from}{\leftarrow}

\newcommand{\reset}{\mathsf{reset}}
\newcommand{\resetop}[1]{\reset(#1)}

\newcommand{\push}{\mathsf{push}}
\newcommand{\pop}{\mathsf{pop}}
\newcommand{\psicopy}{\psi_{\mathsf{copy}}}

\newcommand{\set}[1]{\left\{ #1 \right\}}
\newcommand{\setof}[2]{\set{#1 \; \middle| \; #2}}

\newif\ifstartedinmathmode
\renewcommand*{\st}{%s.t.~
  \relax\ifmmode\startedinmathmodetrue\else\startedinmathmodefalse\fi
  \ifstartedinmathmode{\;\cdot\;}\else{s.t.~}\fi%
}

\newcommand{\wrt}{w.r.t.~}
\newcommand{\cf}{c.f.~}

\newcommand{\card}[1]{|{#1}|}

\newcommand{\PI}[1]{\text{\sc pi}_{#1}}
\newcommand{\project}[2]{\pi_{#1}(#2)}
\newcommand{\tick}[1]{\checkmark_{\!\!#1}}

\newcommand{\stkout}[1]{\ifmmode\text{\sout{\ensuremath{#1}}}\else\sout{#1}\fi}

\DeclareMathAlphabet{\mathcalligra}{T1}{calligra}{m}{n}
\newcommand{\arr}{r}%\scriptr r\,}

\newcommand{\slawek}[1]{\todo[color=blue!30]{{\bf S\l{}awek:} #1}}

\newtheorem{claim}[theorem]{Claim}
\usepackage{microtype}
\usepackage[utf8]{inputenc}
\usepackage{thm-restate}

\usepackage{bbm}
\usepackage{mathabx}
\usepackage{tikz}
\usetikzlibrary{patterns}

\usepackage{xspace}

\usepackage[T1]{fontenc}
\usepackage{wrapfig}

\title{Binary reachability of timed pushdown automata via quantifier elimination and cyclic order atoms}
\titlerunning{Binary reachability of \TPDA via quantifier elimination and cyclic order}

\author{Lorenzo Clemente}
           {University of Warsaw}
           {clementelorenzo@gmail.com}
           {https://orcid.org/0000-0003-0578-9103}
           {Partially supported by Polish NCN grant 2017/26/D/ST6/00201.}
\author{S{\l}awomir Lasota}
           {University of Warsaw}
           {sl@mimuw.edu.pl}
           {https://orcid.org/0000-0001-8674-4470}
           {Partially supported by Polish NCN grant 2016/21/B/ST6/01505.}
\authorrunning{L.~Clemente, S.~Lasota}
\Copyright{Lorenzo Clemente and Sławomir Lasota}

\subjclass{F.1.1 Models of Computation, F.4.1 Mathematical Logic}
\keywords{timed automata, reachability relation, timed pushdown automata, linear arithmetic}

\begin{document}

\maketitle

\begin{abstract}
  We study an expressive model of timed pushdown automata extended with modular and fractional clock constraints.
  We show that the binary reachability relation is effectively expressible in hybrid linear arithmetic with a rational and an integer sort.
  This subsumes analogous expressibility results previously known for finite and pushdown timed automata with untimed stack.
	As key technical tools, we use quantifier elimination for a fragment of hybrid linear arithmetic and for cyclic order atoms,
  and a reduction to register pushdown automata over cyclic order atoms.
  \end{abstract}

% !TEX root = main.tex

\section{Introduction}

%\para{Timed automata with stack.}
%
Timed automata (\TA) are one of the most studied models of reactive timed systems.
%are classical automata equipped with clocks which can be reset and compared by inequality constraints.
%
The fundamental result that paved the way to automatic verification of timed systems
is decidability (and PSPACE-completeness) of the reachability problem for \TA \cite{AD94}.
However, in certain applications, such as in parametric verification,
deciding reachability %for individual pairs of configurations
is insufficient,
and one needs to construct the more general \emph{binary reachability relation},
i.e., the entire (possibly infinite) set of of pairs of configurations $(c_i, c_f)$
s.t.~there is an execution from $c_i$ to $c_f$.
%
%Computing the reachability relation is a theoretically appealing, non-trivial generalisation of the reachability problem.
%Moreover, it has interesting applications, for instance in the context of parametric verification,
%where merely deciding reachability is insufficient.
%
The reachability relation for \TA has been shown to be effectively expressible in hybrid linear arithmetic with rational and integer sorts
\cite{ComonJurski:TA:1999,Dima:Reach:TA:LICS02,KrcalPelanek:TM:FSTTCS:2005,QuaasShirmohammadiWorrell:LICS:2017}.
Since hybrid logic is decidable,
this yields an alternative proof of decidability of the reachability problem.

In this paper, we compute the reachability relation for timed automata extended with a stack.
An early model of  \emph{pushdown timed automata} (\PDTA) extending \TA with a (classical, untimed) stack
has been considered by Bouajjani \emph{et al.}~\cite{bouajjani:timed:PDA:94}.
More recently, \emph{dense-timed pushdown automata} (\dtPDA)
have been proposed by Abdulla \emph{et al.}~\cite{AbdullaAtigStenman:DensePDA:12}
as an extension of \PDTA.
In \dtPDA, stack symbols are equipped with rational \emph{ages},
which initially are $0$ and increase with the elapse of time at the same rate as global clocks;
when a symbol is popped, its age is tested for membership in an interval.
While \dtPDA syntactically extend \PDTA by considering a timed stack,
timed constraints can in fact be removed while preserving the timed language recognised by the \dtPDA,
%as a consequence of the interplay between heavy syntactic restrictions on constraints and the monotonicity of time,
and thus they semantically collapse to \PDTA~\cite{ClementeLasota:LICS:2015}.
%
% by adopting reset-point semantics (c.f.~\cite{Dima:Reach:TA:LICS02}):
%a register stores the absolute time of the last reset of a clock, instead of its current age.
%
This motivates the quest for a strictly more expressive generalisation of \PDTA and \dtPDA with a truly timed stack.
It has been observed in \cite{UezatoMinamide:LPAR15} that adding fractional stack constraints
prevents the stack from being untimed,
and thus strictly enriches the expressive power%
\footnote{%
For \TA, fractional constraints can be handled by the original region construction
and do not make the model harder to analyse \cite{AD94}.
}.

We embrace this observation
and propose the model of \emph{timed pushdown automata} (\TPDA),
which extends timed automata with a timed stack
and integer, fractional, and modulo diagonal/non-diagonal constraints.
%as an expressive extension of timed and pushdown automata with a timed stack.
%
The model features local clocks and stack clocks.
As time elapses, all clocks increase their values,
and they do so at the same rate.
Local clocks can be reset and compared
according to the generalised constraints above.
%
%(This contrasts to \TA,
%where in the presence of epsilon transitions epsilon transitions~\cite{BerardDiekertGastinPetit:1998:Epsilon},
%adding fractional/modular constraints does not increase the expressive power of the model~\cite{ChoffrutGoldwurm:Periodic:2000}.)
%
At the time of a push operation, new stack clocks are created
whose values are initialised, possibly non-deterministically,
as to satisfy a given push constraint between stack clocks and local clocks;
similarly, a pop operation requires that stack clocks to be popped
satisfy a given pop constraint of analogous form.
Stack push/pop constraints are also of the form of diagonal/non-diagonal integer, modulo, and fractional constraints.

\vspace{-3mm}
\para{Contributions.}
We compute the \emph{binary reachability relation} of \TPDA,
i.e., the family of binary relations
$\set {\reach {}  {\ell\arr}} \subseteq \Qgeq^X \times \Qgeq^X$ for control locations $\ell,\arr$
s.t.~from the initial clock valuation $\mu \in \Qgeq^X$ and control location $\ell$
we can reach the final clock valuation $\nu \in \Qgeq^X$ and control location $\arr$,
written $\mu \reach {}  {\ell\arr} \nu$.
The stack is empty at the beginning and at the end of the computation.
The main contribution of the paper is the effective computation of the \TPDA reachability relation
in the existential fragment of \emph{linear arithmetic} \logic, a two-sorted logic combining
Presburger arithmetic $(\Z, \leq, \eqvs, +, 0)$ and linear rational arithmetic $(\Q, \leq, +, 0)$.
As a byproduct of our constructions,
we actually characterise the more general \emph{ternary reachability relation}
$\mu \reach \pi  {\ell\arr} \nu$,
where $\mu, \nu$ are as above
and $\pi : \N^\Sigma$ additionally counts the number of occurrences
of input letters over a finite alphabet $\Sigma$,
i.e., the Parikh image of the run.
%
%Thus, \TPDA are extended with (untimed) input letters and can be used as recognisers of untimed languages.
%
To our knowledge, the ternary reachability relation was not previously considered.
As an application of ternary reachability,
we can model, for instance,
letter counts of initial and final, possibly non-empty, stack contents.
%
%the Parikh image of initial and final stack contents:
%
%In order to go from stack $u$ to stack $v$ (topmost symbol on the right),
%the automaton can initially read $u$ from the input and push it on the stack,
%and then, at the end, can read $v^R$ from the input and pop $v$ from the stack.
%
%Thus, stack counts reduce to input counts,
%and we can use ternary reachability to even relate
%letter counts of initial and final, possibly non-empty, stack contents.
Thus, ternary reachability is an expressive extension of binary reachability.

The computation of the ternary reachability relation is achieved by two consecutive translations.
%each of them allowing for recovering the original reachability relation.
%
%The difficulty of the analysis of dense time is that it combines discrete and continuous features.
%
%We split dense time into its integer and fractional components.
%The integer component is removed by reducing it to the recognised (untimed) language.
%The fractional component is addressed by reduction to cyclic order, which is easier to analyse.
%
First, we transform a \TPDA into a \emph{fractional \TPDA}, which uses only fractional constraints.
In this step we exploit \emph{quantifier elimination} for a fragment of linear arithmetic corresponding to clock constraints.
Quantifier elimination is a pivotal tool in this work,
and to our knowledge its use in the study of timed models is novel.
%
%Moreover, we also remove integer and modular constraints;
The final integer value of clocks is reconstructed by letting the automaton input special tick symbol $\tick x$
every time clock $x$ reaches an integer value (provided it is not reset anymore later);
it is here that ternary reachability is more suitable than binary reachability.

Secondly, a fractional \TPDA is transformed into a \PDA with registers (\RPDA)
over the so called \emph{cyclic order atoms} $(\Q \cap [0, 1), K)$ \cite{ClementeLasota:CSL:2015},
where $K$ is the ternary cyclic order relation %defined as
\begin{align}
	\label{eq:cyc}
  K(a, b, c) \equiv a < b < c \vee b < c < a \vee c < a < b, \qquad \textrm{ for } a, b, c \in \Q \cap [0, 1).
\end{align}
In other words, $K(a, b, c)$ holds if, distributing $a, b, c$ on the unit circle
and going clockwise from $a$, then we fist visit $b$ and afterwards $c$.
Since fractional values are wrapped around $0$ when time increases,
$K$ is invariant under time elapse.
We use registers to store the fractional parts of absolute times of last clock resets;
fractional constraints on clocks are simulated by constraints on registers using $K$.
In order to compute the reachability relation for \RPDA
we use again quantifier elimination,
this time over cyclic order atoms.
The latter property holds since cyclic order atoms constitute a homogeneous structure \cite{survey}.
Therefore, another contribution of this work is the solution of a nontrivial problem such as computing the reachability relation for \TPDA,
which is a clock model,
as an application of \RPDA, which is a register model.
The analysis of \RPDA is substantially easier than a direct analysis of (fractional) \TPDA.
%
%Clearly, if stack and modular constraints are ignored, we get an even simpler proof for classical \TA.
%

From the complexity standpoint, the formula characterising the reachability relation of a \TPDA is computable in double exponential time.
However, when cast down to \TA or \TPDA with timeless stack (which subsume \PDTA and, a posteriori, \dtPDA),
the complexity drops to singly exponential,
matching the previously known complexity for \TA \cite{QuaasShirmohammadiWorrell:LICS:2017}.
For \PDTA, no complexity was previously given in \cite{Dang:PTA:2003},
and thus the result is new.
For \TPDA, the binary reachability problem has not been studied before.
Since the existential fragment of \logic is decidable in NP
(because so is existential linear rational arithmetic \cite{Sontag:IPL:1985}
and existential Presburger arithmetic \cite{Weispfenning:JSC:1988}),
we can solve the reachability problem of \TPDA in 2NEXP by reduction to satisfiability for \logic.
%
%Finally, while \TPDA recognise more timed languages than \PDTA,
Since our constructions preserve the languages of all the models involved,
untimed \TPDA languages are context-free.

\para{Discussion.}
From a syntactic point of view,
\TPDA significantly lifts the restrictions of \dtPDA---%
which allow only classical non-diagonal constraints, i.e., interval tests,
and thus has neither diagonal, nor modulo, nor fractional constraints---%
and of the model of \cite{UezatoMinamide:LPAR15}---%
which additionally allows diagonal/non-diagonal fractional tests,
and thus does not have modulo constraints.
Since classical diagonal constraints reduce to classical non-diagonal constraints,
and, \emph{in the presence of fractional constraints},
integer and modulo constraints can be removed altogether (cf.~Sec.~\ref{sec:fractional}),
\TPDA are expressively equivalent to \cite{UezatoMinamide:LPAR15}.
However, while \cite{UezatoMinamide:LPAR15} solves the control state reachability problem,
we solve the more general problem of computing the binary reachability relation.
Our reduction technique not only preserves reachability, like \cite{UezatoMinamide:LPAR15},
but additionally enables the reconstruction of the reachability relation.

Our expressivity result generalises analogous results for \TA
\cite{ComonJurski:TA:1999,Dima:Reach:TA:LICS02,KrcalPelanek:TM:FSTTCS:2005,QuaasShirmohammadiWorrell:LICS:2017}
and \PDTA \cite{Dang:PTA:2003}.
%in discrete \cite{DangEtAl:CAV:2000} and dense time \cite{Dang:PTA:2003}.
%
%While the results of these works are similar in spirit, their techniques are different
%and do not immediately generalise to our setting.
%
The proof of \cite{ComonJurski:TA:1999} for \TA has high technical difficulty
and does not yield complexity bounds.
The proof of \cite{Dima:Reach:TA:LICS02} for \TA uses an automata representation
for sets of clock valuations;
the idea of \emph{reset-point semantics} employed in \cite{Dima:Reach:TA:LICS02}
is analogous to using registers instead of clocks.
The paper \cite{KrcalPelanek:TM:FSTTCS:2005} elegantly expresses the reachability relation for \TA
with \emph{clock difference relations} (CDR) over the fractional values of clocks.
It is remarkable that the formulas expressing the reachability relations
that we obtain are of the same shape as CDR.
The recent paper \cite{QuaasShirmohammadiWorrell:LICS:2017}
shows that the \TA binary reachability relation can be expressed
in the same fragment of hybrid linear arithmetic that we use for \TPDA,
which we find very intriguing.
Their proof converts the integer value of clocks into counters,
and then observes that, thanks to the specific reset policy of clocks,
these counter machines have a semilinear reachability relation;
the latter is proved by encoding the value of counters into the language.
In our proof, we bring the encoding of the integer value of clocks into the language
to the forefront, via the introduction of the ternary reachability relation.
The proof of \cite{Dang:PTA:2003} for \PDTA
also separates clocks into their integer and fractional part.
It is not clear how any of the previous approaches
%\cite{Dima:Reach:TA:LICS02,KrcalPelanek:TM:FSTTCS:2005,QuaasShirmohammadiWorrell:LICS:2017Dang:PTA:2003}
could handle a timed stack.

Another approach for computing the reachability relation for \TPDA
would be to reduce it directly to a more expressive register model,
such as \emph{timed register pushdown automata} (\TRPDA)
\cite{ClementeLasota:LICS:2015,ClementeLasotaLazicMazowiecki:LICS:2017},
which considers both integer $(\Z, \leq, +1)$
and rational registers $(\Qgeq, \leq)$.
While such a reduction for the reachability problem is possible
since (the integer part of) large clock values can be ``forgotten'',
e.g., along the lines of \cite{ClementeLasota:LICS:2015},
this does not hold anymore if we want to preserve the reachability relation.
For this reason, in the present work we first remove the integer part of clocks
(by encoding it in the untimed language)
and then we reduce to \RPDA, which have only fractional registers and no integer register,
and are thus easier to analyse than \TRPDA%
\footnote{%
\TRPDA are more general than \RPDA---%
cyclic order atoms can be interpreted into $(\Qgeq, \leq)$.
The binary reachability relation for \TRPDA
can be computed by refining the reductions of \cite{ClementeLasotaLazicMazowiecki:LICS:2017}
used for deciding the reachability problem.
However, we do not know how to use the reachability relation of \TRPDA to compute that of \TPDA.
}.
The method of quantifier elimination was recently applied to the analysis of another timed model,
namely \emph{timed communicating automata} \cite{Clemente:TCA:ArXiv:2018}.
%One direction, inspired by modelling time by means of registers instead of clocks~\cite{BL12icalp}
%

Finally, another expressive extension of \TA, called \emph{recursive timed automata} (\RTA),
has been proposed \cite{TrivediWojtczak:RTA:2010,BenerecettiMinopoliPeron:RTA:2010}.
\RTA use a timed stack to store the current clock valuation,
which does not evolve as time elapses
and can be restored at the time of pop.
This facility makes \RTA expressively incomparable to all models previously mentioned.

\para{Notations.}
%
%Throughout the paper, we use the following notations.
Let $\Q$, $\Qgeq$, $\Z$, and $\N$ denote the rationals, the non-negative rationals, the integers, and the natural numbers;
let $\unitint = \Qgeq \cap [0, 1)$ be the unit rational interval.
Let $\eqv m$ denote the congruence modulo $m \in \N\setminus\set 0$ in $\Z$.
For $a\in \Q$, let $\floor a \in \Z$ denote the largest integer $k$ \st $k \leq a$,
and let $\fract a = a - \floor a$ denote its fractional part.
Let $\condone C$, for a condition $C$, be $1$ if $C$ holds, and $0$ otherwise.

% !TEX root = main.tex

\section{Linear arithmetic and quantifier elimination}
\label{sec:qe}

%%
%For $a, b, c \in \Q$, let $\floor a \in \Z$ denote the largest integer $k$ \st $k \leq a$,
%let $\fract a = a - \floor a$ denote its fractional part,
%let \emph{cyclic addition} and \emph{difference}
%be $a \cycplus b = \fract {a + b}$, resp., $a \cycminus b := \fract{a - b}$,
%and finally define the \emph{cyclic order} $K(a, b, c)$ as %\lorenzo{at this point, this is a shortcut for a longer formula;}
%%
%\begin{gather}
%	\label{eq:cyc}
%	K(a, b, c) \ \equiv\ a \leq b \leq c \, \vee \, b \leq c \leq a \, \vee \, c \leq a \leq b.
%\end{gather}
%%
%\ignore{
%\islawek{we can remove K and use only this:
%\[
%a \cycminus b = a - b  + \condone{a<b}
%\]
%}
%A useful property relating $\cycminus$ and $K$ is the following
%\begin{gather}
%	\label{eq:cyc:property}
%	a \cycminus b \, + \, b \cycminus c = a \cycminus c + \condone {\neg K(c, b, a)}.
%\end{gather}
%\islawek{thus:
%\[
%a \cycminus b \, + \, b \cycminus c + \condone{a<c} = a \cycminus c + \condone{a<b} + \condone{b<c}
%\]
%}
%}
%%
Consider the two-sorted structure $\A = \AZ \uplus \AQ$,
where $\AZ = (\Z, \leq, \eqvs, +, (k)_{k\in\Z})$
and $\AQ = (\Q, \leq, +, (k)_{k\in \Q})$.
We consider ``$+$'' as a binary function, and we have a constant $k$ for every integer/rational number.
By \emph{linear arithmetic}, denoted \logic, we mean the two-sorted first-order language in the vocabulary of $\A$.
%as the first-order theory $\L_{\Z,\Q}$
Restriction to the integer sort yields Presburger arithmetic $\logicint$ (\emph{integer formulas}),
and restriction to the rational sort yields linear rational arithmetic $\logicrat$ (\emph{rational formulas}).
We assume constants are encoded in binary.

Two formulas are \emph{equivalent} if they are satisfied by the same valuations.
It is well-known that the theories of $\AZ$~\cite{Presburger:1930} and $\AQ$~\cite{FerranteRackoff:QE:Reals}
admit effective elimination of quantifiers:
Every formula can effectively be transformed in an equivalent quantifier-free one.
Therefore, the theory of $\A$ also admits quantifier elimination, by the virtue
of the following general fact (when speaking of a structure admitting quantifier elimination, we have in mind its theory).
% (SL camera ready)  The proofs for this section can be found in Sec.~\ref{app:qe}.

\begin{restatable}{lemma}{qeUnion}
	\label{lem:qe-union}
	If the structures $\str_1$ and $\str_2$ admit (effective) elimination of quantifiers,
	then the two-sorted structure $\str_1 \uplus \str_2$ also does so.
	For conjunctive formulas, the complexity is the maximum of the two complexities.
\end{restatable}
For clock constraints,
we will use the first-order language
over the two sorted structure $\Aclock = \AZc \uplus \AIc$,
where the integer sort is restricted to $\AZc = (\N, \leq, \eqvs, +1, 0)$---%
the domain is now $\N$ and full addition ``$+$'' is replaced by the unary successor operation ``$+1$'')---%
and the rational sort to $\AIc = (\unitint, \leq, 0)$---
the domain is now the unit interval, there is no addition,
and the only constant is $0$.
Let $\logicclock$ be such a sub-logic.
(As syntactic sugar we allow to use addition of arbitrary, even negative, integer constants in integer formulas, e.g.~$x-4 \leq y+2$.)
As before, $\logicclockint$ and $\logicclockrat$ are the restrictions to the respective sorts.
All the sub-logics above admit effective elimination of quantifiers.
\begin{restatable}{lemma}{qeIntRat}% [Quantifier elimination for $\L_\Z^*$ (cf. \cite{To:CSL:2009})]
	\label{lem:qe-clocks-int-rat}
	The structures % $\AZsucczero$
	$\AZc$ and $\AIc$ admit effective elimination of quantifiers.
	For $\AZc$ the complexity is singly exponential for conjunctive formulas,
	while for $\AIc$ is quadratic.
\end{restatable}
%
%\begin{remark}
%	First-order logic of the structure $\AZc$
\noindent
Notice that since $\logicclockint$ is a fragment of Presburger arithmetic $\logicint$,
we could apply the quantifier elimination for $\logicint$
to get a quantifier-free $\logicint$ formula.
Our result is stronger since we get a quantifier-free formula
of the more restrictive fragment $\logicclockint$.
%\end{remark}
%
%\noindent
%The structure
%$\AIc$ admits quantifier elimination too, which together with the last two lemmas
%Lemmas~\ref{lem:qe-union}, \ref{lem:qe-clocks-int}, and \ref{lem:qe-clocks-rat}
%yields:
\begin{corollary}
	\label{cor:qe-clocks}
	The structure $\Aclock$
	admits effective quantifier elimination.
	The complexity is exponential for conjunctive formulas.
\end{corollary}
%

% !TEX root = main.tex

\section{Timed pushdown automata}
\label{sec:tpda}

\subparagraph{Clock constraints.}
Let $X$ be a finite set of clocks.
We consider constraints which can separately speak about the integer $\floor x$ and fractional value $\fract x$
of a clock $x \in X$.
%
%For the integer part we allow diagonal $\floor x - \floor y \sim k$
%and modulo $\floor x - \floor y \eqv m k$ constraints,
%as well as their non-diagonal counterparts $\floor x \sim k$ and $\floor x \eqv m k$.
%
%For the fractional part we allow order constraints $\fract x \leq \fract y$
%as well as testing for zero $\fract x = 0$.
%
%
A \emph{clock constraint} over $X$ is a boolean combination of
\emph{atomic clock constraints} of one of the forms
\begin{align*}
%	\begin{tabular}{}
	&
%		&\textrm{(class}&\textrm{ical)}
		&\textrm{(inte}&\textrm{ger)}
		&\textrm{(modu}&\textrm{lar)}
		&\textrm{(fracti}&\textrm{onal)}
	\\[1ex]
	&\textrm{(non-diagonal)}	%&x		&\leq k	\ \
		&\floor x & \leq k 	\ \  			&\floor x & \eqv m k  	\ \   	 	&\fract x & = 0   \\
	&\textrm{(diagonal)}	%&x - y &\leq k
		&\floor x - \floor y & \leq k	&\floor x - \floor y & \eqv m k		&\fract x & \leq \fract y
%\end{tabular}
\end{align*}
where
$x, y \in X$, $m \in \N$,a and $k \in \Z$. %, and $\sim \ \in \{ <, \leq, \geq, > \}$.
%The ones in the lower row we call \emph{diagonal} and the other ones \emph{non-diagonal};
%the ones in the first column we call \emph{classical},
%in the second column \emph{integer},
%in the third column \emph{modular},
%and in the last column \emph{fractional}.
%
Since we allow arbitrary boolean combinations,
we consider also the constraint $\true$, which is always satisfied,
and variants with any $\sim \ \in \{\leq, <, \geq, >\}$ in place of $\leq$.
A \emph{clock valuation} %over a set of clocks $X$
is a mapping $\mu \in \Qgeq^X$ assigning a non-negative rational number to every clock in $X$;
we write $\floor \mu$ for the valuation in $\N^X$ \st $\floor \mu(x) := \floor {\mu(x)}$
and $\fract \mu$ for the valuation in $\I^X$ \st $\fract \mu(x) := \fract {\mu(x)}$.
For a valuation $\mu$ and a clock constraint $\varphi$ %over clocks $x_1, \dots, x_k$, we write $\mu \models \varphi$
we say that $\mu$ \emph{satisfies} $\varphi$
if $\varphi$ is satisfied when integer clock values $\floor x$ are evaluated according to $\floor \mu$
and fractional values $\fract x$ according to $\fract \mu$.

%
%
%For convenience we include disjunction $\vee$, even though it can be eliminated
%by using nondeterminism in the transition relation of the automaton.
%
%For an atomic constraint $\psi$ and $\varphi$ in disjunctive normal form (DNF), \slawek{do we need this?}
%we write $\psi \in \varphi$ to denote that $\psi$ is a conjunct of $\varphi$.
%
%A \emph{clock formula} is a first-order formula where atomic formulas are (atomic) clock constraints.
%
\begin{remark}[Clock constraints as quantifier-free $\logicclock$ formulas]
	\label{rem:constraints-as-qf-formulas}
	Up to syntactic sugar,
	a clock constraint over clocks $\{x_1, \ldots, x_n\}$ is the same as a quantifier-free $\logicclock$ formula
	$\varphi(\floor {x_1}, \ldots, \floor {x_n}, \fract {x_1},  \ldots, \fract {x_n})$
	over $n$ integer and $n$ rationals variables.
	%, with the implicit condition that integer parts are nonnegative and fractional parts belong to the rational interval $[0, 1)$.
\end{remark}
\begin{remark}[Classical clock constraints]
	\label{rem:sugar}
	Integer and fractional constraints subsume classical ones.
	For clocks $x, y$, since $x = \floor x + \fract x$ (and similarly for $y$)%
	\footnote{We often identify a clock $x$ with its value for simplicity of notation.},
	$x - y \leq k$ for an integer $k$ is equivalent to
	$(\floor x - \floor y \leq k \wedge \fract x \leq \fract y) \vee \floor x - \floor y \leq k - 1$,
	and similarly for $x\leq k$.
	On the other hand, the fractional constraint $\fract x = 0$
	is not expressible as a classical constraint.
%
%	We can also express modulo constraints of the form $x - y \eqv m k$ (meaning
%	that $x - y$ is an integer, and this integer is equivalent to $k$ modulo $m$)
%	as a combination of integer modulo and fractional constraints: % thanks to the equivalence
%	%
%	\begin{gather*}
%		x - y \eqv m k \quad \equiv \quad
%		\fract x = \fract y \wedge \floor x - \floor y \eqv m k.
%	\end{gather*}
\end{remark}
\ignore{
	\begin{remark}[Invariance under time elapse]
		\label{rem:invariance}
		While classical diagonal constraints $x - y \leq k$ are invariant \wrt the elapse of time,
		this is the case neither for integer $\floor x - \floor y \leq k$ nor for fractional constraints $\fract x \leq \fract y$.
		In Sec.~\ref{sec:fractional}, it will be useful to separate a constraint into its invariant and non-invariant part,
		which is achieved by replacing an integer constraint $\floor x - \floor y \leq k$
		by the equivalent $x - y \leq k \;\vee\; (x - y \leq k + 1 \wedge \fract x > \fract y)$.
	%
	%	Finally, it is interesting to note that while fractional constraints $\fract x \leq \fract y$ are not invariant \wrt the elapse of time,
	%	the ternary relation of their \emph{cyclic order} $K$ is invariant,
	%	where $K(\fract x, \fract y, \fract z)$ holds iff
	%	\begin{align*}
	%		 \fract x < \fract y < \fract z \vee \fract z < \fract x < \fract y \vee \fract y < \fract z < \fract x.
	%	\end{align*}
	%
	%	We won't need this observation in the rest of the paper.
	\end{remark}
}
\begin{remark}[$\floor x - \floor y$ versus $\floor {x - y}$]
	\label{rem:int-diff}
	In the presence of fractional constraints,
	the expressive power would not change if,
	instead of atomic constraints $\floor x - \floor y \eqv m k$ and $\floor x - \floor y \leq k$
	speaking of the \emph{difference of the integer parts},
	we would choose $\floor{x - y} \eqv m k$ and  $\floor{x-y} \leq k$
	speaking of the \emph{integer part of the difference},
	since the two are inter-expressible:
	%(where $1^?$ is $1$ if $\fract x < \fract y$ and $0$ otherwise)
%
	\begin{align}
		\label{eq:floor:1}
		\floor {x - y} = \floor x - \floor y - \condone {\fract x < \fract y} \qquad  \textrm{ and } \qquad
		%
% not needed in the rest
%		\label{eq:floor:2}
%		\floor x - \floor y &= \left\{\begin{array}{ll}
%			\floor {x - y} 		& \textrm{ if } \fract x \geq \fract y \\
%			\floor {x - y} + 1	& \textrm{ otherwise}.
%		\end{array}\right. \\
		%
%		\label{eq:frac:1}
		\fract {x - y} &= \fract x - \fract y + \condone {\fract x < \fract y}.
		%= \fract x \cycminus \fract y
	\end{align}
\end{remark}

%i.e.,
%if $[\floor {x_1} \mapsto \floor {\mu(x_1)}, \cdots, \floor {x_k} \mapsto \floor {\mu(x_k)}, \fract {x_1} \mapsto \fract {\mu(x_1)}, \cdots, \fract {x_k} \mapsto \fract {\mu(x_k)} ] \models \varphi$,
%
%Two clock constraints $\varphi$ and $\psi$ are \emph{equivalent} if $\sem{\varphi} = \sem{\psi}$.
%

\subparagraph{The model.}

%\cite{AbdullaAtigStenman:DensePDA:12}
A \emph{timed pushdown automaton} (\TPDA) is a tuple
$\P = \tuple {\Sigma,\Gamma, L, X, Z, \Delta}$
where
$\Sigma$ is a finite input alphabet,
$\Gamma$ is a finite stack alphabet,
$L$ is a finite set of control locations,
$X$ is a finite set of \emph{global clocks}, %containing a special clock $x_0$ which is never reset,
and $Z$ is a finite set of \emph{stack clocks} disjoint from $X$.
The last item $\Delta$ is a set of transition rules $\TPDArule \ell \op \arr$ with $\ell, \arr \in L$ control locations,
where $\op$ determines the type of transition:
\begin{itemize}
\item \emph{time elapse} $\op = \elapse$,
\item \emph{input} $\op = a \in \Sigma_\varepsilon := \Sigma \cup \{\varepsilon\}$ an input letter,
\item \emph{test} $\op = \varphi$ a \emph{transition constraint} over clocks $X$,
\item \emph{reset} $\op = \resetop Y$ with $Y \subseteq X$ a set of clocks to be reset,
\item \emph{push} $\op = \pushopgen{\alpha}{\psi}$ with $\alpha \in \Gamma$ a stack symbol to be pushed on the stack
under the \emph{stack constraint} $\psi$ over clocks $X \cup Z$, or
\item \emph{pop} $\op = \popopgen{\alpha}{\psi}$ similarly as push.
\end{itemize}
We assume that every atomic constraint in a stack constraint contains some stack variable from $Z$.
%otherwise they can be checked as a transition constraint.
Throughout the paper, let $x_0$ be a global clock that is never reset
(and thus measures the total elapsed time),
and let $z_0$ be a stack clock that is $0$ when pushed.
A \TPDA has \emph{untimed stack} if the only stack constraint is $\true$.
%in which case we omit the constraint in stack operations and write simply $\pushop \alpha$ and $\popop \alpha$.
%
Without push/pop operations, we obtain nondeterministic timed automata (\TA).

\begin{remark}[Complexity]
	For complexity estimations, we assume that constraints are conjunctions of atomic constraints,
	that constants therein are encoded in binary,
	that $M$ is the maximal constant,
	and that all modular constraints use the same modulus $M$.
\end{remark}

\begin{remark}[Time elapse]
	The standard semantics of timed automata where time can elapse freely in every control location
	is simulated by adding explicit time elapse transitions
	$\TPDArule{\ell}{\elapse}{\ell}$ for suitable locations $\ell$.
	Our explicit modelling of the elapse of time will simplify the constructions in Sec.~\ref{sec:fractional}.
\end{remark}

\begin{remark}[Comparison with \dtPDA] % \cite{AbdullaAtigStenman:DensePDA:12}]
	The \dtPDA model \cite{AbdullaAtigStenman:DensePDA:12} allows only one stack clock $Z = \set z$
	and stack constraints of the form $z \sim k$.
	As shown in~\cite{ClementeLasota:LICS:2015}, this model is equivalent to \TPDA with untimed stack.
	Our extension is two-fold.
	First, our definition of stack constraint is more liberal,
	since we allow more general \emph{diagonal stack constraints} of the form $z - x \sim k$.
	%however as a byproduct of Sec.~\ref{sec:bounded} this extension still reduces to \TPDA with untimed stack.
	%
	Second, we also allow \emph{modular}
	$\floor y - \floor x \eqv m k$ and \emph{fractional constraints} $\fract x \sim \fract y$, where clocks $x, y$
	can be either global or stack clocks.
	As demonstrated in Example~\ref{ex:palin}
	below, this model is not reducible to untimed stack,
	and thus \TPDA are more expressive than \dtPDA.
\end{remark}

\subparagraph{Semantics.}
%
%The formal semantics of \TPDA follows~\cite{AbdullaAtigStenman:DensePDA:12}.
%
Every stack symbol is equipped with a fresh copy of clocks from $Z$.
At the time of $\pushopgen{\alpha}{\psi}$,
the push constraint $\psi$ specifies possibly nondeterministically the initial value of all clocks in $Z$ \wrt global clocks in $X$.
Both global and stack clocks evolve at the same rate when a time elapse transition is executed.
At the time of $\popopgen{\alpha}{\psi}$,
the pop constraint $\psi$ specifies the final value of all clocks in $Z$ \wrt global clocks in $X$.
%
%Two \TPDA are \emph{equivalent} if they recognise the same timed language.
%
A \emph{timed stack} is a sequence $w \in (\Gamma \times \Qgeq^Z)^*$ of pairs $(\gamma, \mu)$,
where $\gamma$ is a stack symbol and $\mu$ is a valuation for stack clocks in $Z$.
For a clock valuation $\mu$ and a set of clocks $Y$, let $\mu[Y \mapsto 0]$ be the same as $\mu$ except that clocks in $Y$ are mapped to $0$.
For $\delta \in \Qgeq$, let $\mu + \delta$ be the clock valuation which adds $\delta$ to the value of every clock,
i.e., $(\mu + \delta)(x) := \mu(x) + \delta$,
and for a timed stack $w = (\gamma_1, \mu_1) \cdots (\gamma_k, \mu_k)$,
let $w + \delta$ be $(\gamma_1, \mu_1 + \delta) \cdots (\gamma_k, \mu_k + \delta)$.
A \emph{configuration} is a triple $\tuple{\ell, \mu, w} \in L \times \Qgeq^X \times (\Gamma \times \Qgeq^Z)^*$
where $\ell$ is a control location, $\mu$ is a clock valuation over the global clocks $X$, and $w$ is a timed stack.
Let $\tuple{\ell, \mu, u}, \tuple{\arr, \nu, v}$ be two configurations.
For every input symbol or time increment $a \in (\Sigma_\varepsilon \cup \Qgeq)$ we have a transition
%
%\begin{gather*}
$\tuple{\ell, \mu, u} \goesto a \tuple{\arr, \nu, v}$
%\end{gather*}
%
whenever there exists a rule $\TPDArule{\ell}{\op}{\arr} \in \Delta$ \st one of the following holds:
%
%\begin{multicols}{2}
\begin{itemize}
	\item%[\bf (elapse)]
	$\op = \elapse$,
	$a \in \Qgeq$, $\nu = \mu + a$, $v = u + a$.
	\item%[\bf (input)]
	$\op = a \in \Sigma_\varepsilon$, $\nu = \mu$, $u = v$.
	\item%[\bf (test)]
	$\op = \varphi$, $a = \varepsilon$, $\mu \models \varphi$, $\nu = \mu$, $u = v$.
	\item%[\bf (reset)]
	$\op = \resetop Y$, $a = \varepsilon$, $\nu = \mu[Y \mapsto 0]$, $v = u$.
	\item%[\bf (push)]
	$\op = \pushopgen{\gamma}{\psi}$, $a = \varepsilon$, $\mu = \nu$,
	$v = u \cdot \tuple{\gamma, \mu_1}$
	if $\mu_1 \in \Qgeq^Z$ satisfies $(\mu, \mu_1) \models \psi$,
	where $(\mu, \mu_1) \in \Qgeq^{X \cup Z}$ is the unique clock valuation
	that agrees with $\mu$ on $X$ and with $\mu_1$ on $Z$.
	\item%[\bf (pop)]
	$\op = \popopgen{\gamma}{\psi}$, $a = \varepsilon$, $\mu = \nu$,
	$u = v \cdot \tuple{\gamma, \mu_1}$
	provided that $\mu_1 \in \Qgeq^Z$ satisfies $(\mu, \mu_1) \models \psi$.
\end{itemize}
%\end{multicols}
%
A \emph{timed word} is a sequence ${w = \delta_1 a_1 \cdots \delta_n a_n \in (\Qgeq\Sigma_\varepsilon)^*}$
of alternating time elapses and input symbols;
the one-step transition relation $\tuple{\ell, \mu, u} \goesto a \tuple{\arr, \nu, v}$
is extended on timed words $w$ as $\tuple{\ell, \mu, u} \goesto w \tuple{\arr, \nu, v}$
in the natural way.
The \emph{timed language} from location $\ell$ to $\arr$ is
$\lang {\ell, \arr} := \setof
	{\project {\varepsilon} w \in (\Qgeq\Sigma)^*}
	{\tuple{\ell, \mu_0, \varepsilon} \goesto w \tuple{\arr, \mu_0, \varepsilon}}$
where $\project {\varepsilon} w$ %$\cdot : (\Qgeq\Sigma_\varepsilon)^* \to (\Qgeq\Sigma)^*$
removes the $\varepsilon$'s from $w$ and
$\mu_0$ is the valuation that assigns $\mu_0(x) = 0$ to every clock $x$.
The corresponding \emph{untimed language} $\languntimed {\ell, \arr}$
is obtained by removing the time elapses from $\lang {\ell, \arr}$.
%
%Two \TPDA $\P$ and $\P'$ are
%\slawek{(1) equivalence is def. in asymmetric way; (2) should we care about the languages, or rather about rechability relation?}
%\emph{equivalent} if for every pair of states $\ell,\arr$ of $\P$
%there exists a pair of states $\ell',\arr'$ of $\P'$ \st $L_\P(\ell, \arr) = L_{\P'}(\ell', \arr')$.

\begin{example} \label{ex:palin}
	Let $L$ be the timed language of even length palindromes
	\st the time distance between every pair of matching symbols is an integer:
	\begin{gather*}
		L = \setof {\delta_1 a_1 \cdots \delta_{2n} a_{2n} \ignore{\in (\Qgeq \Sigma)^*} } { \forall (1 \leq i \leq n) \st a_i = a_{2n - i  + 1} \wedge \delta_{i+1} + \cdots + \delta_{2n - i  + 1} \in \N}.
	\end{gather*}
	$L$ can be recognised by a \TPDA over input and stack alphabet $\Sigma = \Gamma = \set{a, b}$,
	with locations $\ell, \arr$, no global clock, one stack clock $Z = \set z$,
	and the following transition rules (omitting some intermediate states), where $\alpha$ ranges over $\set{a,b}$:
	\begin{align*}
		& \TPDArule{\ell}{\alpha; \pushopgen{\alpha}{\fract{z} = 0}}{\ell} %& (\textrm{for } x \in \set {a, b}) &
		&& \TPDArule{\ell}{\varepsilon}{\arr} \\
		& \TPDArule{\arr}{\alpha; \popopgen{\alpha}{\fract{z} = 0 \ignore{\land \floor z \eqv 2 0}}}{\arr} %& (\textrm{for } x \in \set {a, b}) &
		&& \TPDArule{\ell}{\elapse}{\ell}, \TPDArule{\arr}{\elapse}{\arr}
	\end{align*}
	%
%	(A $\nop$ is a syntactic sugar for a dummy push followed by a pop, using an auxiliary control location.)
	We have $L = L(\ell, \arr)$.
	Since $L$ cannot be recognised by \TPDA with untimed stack
	(cf.~\cite{UezatoMinamide:LPAR15}), %$L_{\textrm{ex}}$
	fractional stack constraints strictly increase the expressive power of the model.
%	Notice also that we could replace the fractional constraint above $\fract {z} = 0$
%	with the equivalent modulo constraint $z \eqv 1 0$,
%	this showing that also stack modulo constraints strictly increase the expressiveness of the model.
\end{example}

\subparagraph{The reachability relation.}
%
%
%The \emph{untiming} of a timed word $w$ is the sequence restricted to $\Sigma$, i.e., the one obtained by replacing time increments by $\varepsilon$.
%
The \emph{Parikh image} of a timed word $w$
is the mapping $\PI w \in \N^\Sigma$
\st $\PI w (a)$ is the number of $a$'s in $w$,
ignoring the elapse of time and $\varepsilon$'s.
For two control locations $\ell, \arr$, clock valuations $\mu, \nu \in \Qgeq^X$,
and a timed word $w \in (\Qgeq \Sigma_\varepsilon)^*$,
we write $\mu \reach w {\ell\arr} \nu$
if $\tuple{\ell, \mu, \varepsilon} \goesto w \tuple{\arr, \nu, \varepsilon}$.
We overload the notation and, for $\pi \in \N^\Sigma$, we write
%
%\begin{gather*}
$
	\mu \reach \pi {\ell\arr} \nu
$
%\end{gather*}
%
if there exists a timed word $w$ \st $\mu \reach w {\ell\arr} \nu$ and $\pi = \PI w$.
We see $\set{\reach {} {\ell\arr}}_{\ell,\arr \in L}$
as a family of subsets of $\Qgeq^X \times \N^\Sigma \times \Qgeq^X$
and we call it the \emph{ternary reachability relation}.

Let $\set {\psi_{\ell\arr}(\floor {\vec x}, \fract {\vec x}, \vec f, \floor {\vec y}, \fract {\vec y})}_{\ell,\arr\,\in L}$
be a family of \logic formulas,
where $\floor {\vec x},\floor {\vec y}$ represent the integer values of initial and final clocks,
$\fract {\vec x},\fract {\vec y}$ their fractional values,
and $\vec f$ letter counts.
The reachability relation $\set{\reach {} {\ell\arr}}_{\ell,\arr\in L}$
is \emph{expressed} by the family of formulas $\set {\psi_{\ell\arr}}_{\ell,\arr \in L}$
if the following holds:
For every control locations $\ell,\arr \in L$,
clock valuations $\mu, \nu \in \Qgeq^X$
and $\pi \in \N^\Sigma$,
$\mu \reach \pi {\ell\arr} \nu$ holds, if, and only if,
$(\floor \mu, \fract \mu, \pi, \floor \nu, \fract \nu) \models \psi_{\ell\arr}$ holds.
%
%In the above statement, the integer parts $\floor \mu$, $\floor \nu$ of clock valuations as well as $f$
%yield values for integer variables $(\floor {\vec x})$, and the fractional parts $\fract \mu$, $\fract \nu$ for rational variables of $\varphi_{\ell\arr}$.

\subparagraph*{Main results.}
As the main result of the paper we show that the reachability relation of \TPDA and \TA
is expressible in linear arithmetic \logic.
\begin{theorem}
	\label{thm:TPDA}
	The reachability relation %$\reach {} {\ell\arr}$
	of a \TPDA is expressed by a family of existential \logic formulas % $\psi_{\ell\arr}$
	computable in double exponential time.
	For \TA, the complexity is % single
	exponential.
\end{theorem}
This is a strengthening of analogous results for \TA \cite{ComonJurski:TA:1999,QuaasShirmohammadiWorrell:LICS:2017}
since our model, even without stack, is more expressive than classical \TA due to fractional constraints.
As a side effect of the proofs we get:
\begin{theorem}
	\label{thm:lang}
	Untimed \TPDA languages $\languntimed {\ell, \arr}$ are effectively context-free.
\end{theorem}
The following two sections are devoted to proving the two theorem above.

% !TEX root = main.tex

\section{Fractional \TPDA}
\label{sec:fractional}

A \TPDA is \emph{fractional} if it contains only fractional constraints.
We show that computing the reachability relation
reduces to the same problem for fractional \TPDA.
%
%The reachability relation of a fractional \TPDA itself will be computed in Sec.~\ref{sec:cfg}
%by constructing a context-free grammar.
Our transformation is done in three steps,
each one further restricting the set of allowed constraints.
\begin{itemize}
	\item[\bf A] The \TPDA is \emph{push-copy}, that is,
	push operations can only copy global clocks into stack clocks.
	There is one stack clock $z_x$ for each global clock $x$,
	and the only push constraint is
	\begin{gather}
		\label{eq:copy:push:constraint}
			\psicopy(\vec x, \vec z_{\vec x}) \ \equiv\ \bigwedge_{x \in X} \floor {z_x} = \floor x \wedge \fract {z_x} = \fract x.
	\end{gather}
	By pushing copies of global clocks into the stack,
	we can postpone checking all non-trivial stack constraints to the time of pop.
	This steps uses quantifier elimination.
	The blowup of the number of pop constraints and stack alphabet is exponential.

	\item[\bf B] The \TPDA is \emph{pop-integer-free}, that is,
	pop transitions do not contain integer constraints.
	The construction is similar to a construction from \cite{ClementeLasota:LICS:2015}
	and is presented in Sec.~\ref{app:pop:integer:free}.
	Removing pop integer constraints is crucial towards removing all integer clocks
	(modulo constraints will be removed by the next step).
	This step strongly relies on the fact that stack clocks are copies of global clocks,
	which allows one to remove integer pop constraints
	by reasoning about analogous constraints between global clocks at the time of push
	and their future values at the time of pop,
	thus bypassing the stack altogether.
	We introduce one global clock for each integer pop constraint,
	exponentially many locations in the number of clocks and pop constraints,
	and exponentially many stack symbols in the number of pop constraints.
	When combined with the previous step,
	altogether exponentially many new clocks are introduced,
	and doubly exponentially many locations/stack symbols.
	It is remarkable that pop integer constraints can be removed
	by translating them into finitely many transition constraints on global clocks.

 	\item[\bf C] The \TPDA is {fractional}.
	All integer clocks are removed.
	In order to recover their values (which are needed to express the reachability relation),
	a special symbol $\tick x$ is produced when an integer clock elapses one time unit.
	This step introduces a further exponential blowup of control locations \wrt global clocks
	and polynomial in the maximal constant $M$.
	The overall complexity of control locations thus stays double exponential.

\end{itemize}
\noindent
By {\bf A}+{\bf B}+{\bf C} (in this order, since the latter properties are ensured assuming the previous ones),
we get the following theorem.
\begin{restatable}{theorem}{thmFractionalTPDA}
	\label{thm:fractional:TPDA}
	%The reachability relation of \TPDA \slawek{shouldn't we explain what it means?}
	%effectively reduces to the reachability relation of fractional \TPDA
	%with a double exponential complexity blowup.
	%
%	The number of control locations has an exponential blowup,
%	and the number of clocks has linear blowup.
	%
	A \TPDA $\P$ can be effectively transformed into a
	%(push-copy, pop-simple)
	fractional \TPDA $\QQ$
	\st	a family of \logic formulas $\set{\varphi_{\ell\arr}}$ expressing the reachability relation of $\P$
	can effectively be computed from a family of \logic formulas $\set{\varphi_{\ell'\arr'}'}$
	expressing the reachability relation of $\QQ$. %\slawek{the complexity is missing of transforming $\varphi'$ into $\varphi$}
	The number of control locations and the size of the stack alphabet in $\QQ$ have a double exponential blowup,
	and the number of clocks has an exponential blowup.
\end{restatable}
\noindent
If there is no stack, then we do not need the first two steps, and we can do directly {\bf C}.
\begin{restatable}{corollary}{corNTA}
	\label{cor:fractional:NTA}
	The reachability relation of push-copy \TPDA/\,\TA
	effectively reduces to the reachability relation of fractional \TPDA/\,\TA
	with an exponential blowup in control locations.
%	The number of control locations has an exponential blowup,
%	and the number of clocks has linear blowup.
%\slawek{the statement is unclear}
%\slawek{in case of push-copy \TPDA, there is also exponential blowup of the size of stack alphabet}
\end{restatable}

\subsection*{\bf (A) The \TPDA is push-copy}
Let $K_\leq$ be the non-strict variant of the ternary cyclic order $K$ from \eqref{eq:cyc},
defined as $K_\leq(a,b,c) \equiv K(a,b,c) \vee a=b \vee b=c$
for $a, b, c \in \unitint$.
%let the (strict) \emph{cyclic order} $K(a, b, c)$,
%and its non-strict variant $K_\leq(a,b,c)$ be:
%
%\begin{align}
%	\label{eq:cyc}
%	\begin{aligned}
%	K(a, b, c) \ & \equiv\ a < b < c \, \vee \, b < c < a \, \vee \, c < a < b
%	\\
%	K_\leq(a,b,c) \ & \equiv \ K(a,b,c) \, \vee \, a=b \, \vee \, b=c.
%	\end{aligned}
%\end{align}
%
%These relations will be useful in this section as well as in Sec.~\ref{sec:frac2reg}
%to describe the fractional parts of clocks.
%
Let $\psi_\push(\vec x, \vec z)$ be a push constraint,
and let $\psi_\pop(\vec x', \vec z')$ be the corresponding pop constraint.
Since stack clock $z_0$ is $0$ when pushed on the stack,
$z_0'$ is the total time elapsed between push and pop;
let $\vec z_0' = (z_0', \dots, z_0')$ (the length of which depends on the context).
%
%\vec z_{\vec x},
Let $\vec z_{\vec x}'$ be a vector of stack variables
representing the value of \emph{global clocks} at the time of pop,
provided they were not reset since the matching push.
% (SL camera ready)   (thus $z'_{x_0} = z_0'$).
%
%Let $\delta \geq 0$ be the time elapsed between push and pop,
%and let $\vec \delta = (\delta, \dots, \delta)$ (the length of which depends on the context).
%
Since all clocks evolve at the same rate,
for every global clock $x$ and stack clock $z$, we have
\begin{align}
	\label{equivalence}
	x = z'_x - z_0' \qquad \textrm{ and } \qquad z = z' - z_0'.
\end{align}
If at the time of push, instead of pushing $\vec z$,
we push on the stack a copy of global clocks $\vec x$,
then at the time of pop it suffices to check that the following formula holds
\begin{align}
	\label{eq:psi:pop'}
	\psi_\pop'(\vec x', \vec z'_{\vec x}) \ \equiv\ \exists \vec z' \geq \vec 0 \st
		\psi_\push(\vec z'_{\vec x} - \vec z_0', \vec z' - \vec z_0') \wedge \psi_\pop(\vec x', \vec z').
\end{align}
%
%where intuitively we guess the total time $\delta$ elapsed between push and pop,
%together with the final value $\vec z'$ of the stack clocks originally pushed. % (other similar guesses would do).
%
Note that the assumption that $z_0 = 0$ at the time of push makes the existential quantification satisfiable by exactly one
value of $z'_0$, namely the total time elapsed between push and pop.
However, $\psi_\push(\vec z'_{\vec x} - \vec z_0', \vec z' - \vec z_0')$ is not a constraint anymore,
since variables are replaced by differences of variables.
We resolve this issue by showing that the latter is in fact equivalent to a clock constraint.
Thanks to \eqref{equivalence},
for every clock $x$ we have
${\floor x = \floor {z'_x - z_0'}}, {\fract x = \fract {z'_x - z_0'}}$, % = z'_x \cycminus z_0'}$,
and ${\floor z = \floor {z' - z_0'}}, {\fract z = \fract {z' - z_0'}}$. % = z' \cycminus z_0'}$.
Thus, a fractional constraint $\fract y \leq \fract z$ in $\psi_\push$
is equivalent to $\fract {z'_y - z_0'} \leq \fract {z' - z_0'}$,
which is in turn equivalent to
$C = K_\leq(\fract {z_0'}, \fract {z'_y}, \fract {z'})$,
which is definable from $\leq$.
%
%Furthermore,
%Since $z_0' \neq z', z'_y$ by assumption,
%we have
%
%\begin{align}
%	\label{eq:mid}
%	\fract{ z'_y - z_0'} \,-\, \fract {z' - z_0'} = \fract{z'_y - z'} - \condone C,
%	\textrm { with } C = K(\fract {z_0'}, \fract {z'_y}, \fract {z'}),
%\end{align}
%
Moreover, %$\floor y - \floor z$ appearing in $\psi_\push$ integer and modular constraints equals
%
%\begin{align*}
%	\!\!\!\!
		$\floor y - \floor z
		= \floor {z'_y - z_0'} - \floor {z' - z_0'}
		= (z'_y - z_0' - \fract{z'_y - z_0'}) - (z' - z_0' - \fract{z' - z_0'}) %= \\
		= (z'_y - z') - \fract{z'_y - z_0'} + \fract{z' - z_0'} %= (\textrm{by \eqref{eq:mid}}) \\
		= (z'_y - z') - \fract{z'_y - z'} + \condone D = \floor {z'_y - z'} + \condone D$,
		with $D = C \wedge \fract{z'_y} \neq \fract{z'}$.
%\end{align*}
%
\ignore{
and thanks to~\eqref{eq:floor:1} %and \eqref{eq:frac:1}
we obtain
(cf. the detailed calculation in Sec.~\ref{app:push-copy:calculation}),
\begin{gather*}
	\floor y - \floor {z_x} = \floor{z'_y - z'_x} %+ 1^?_{\alpha, \beta}
		\qquad \textrm{ and } \qquad
	\fract y - \fract {z_x} = \fract{z'_y - z'_x}. % - 1^?_{\alpha, \beta},
\end{gather*}
%
%where $\alpha \equiv \fract{z'_y} \geq \fract \delta \wedge \fract{z'_x} > \fract \delta$,
%$\beta \equiv \fract{z'_y} > \fract \delta \wedge \fract{z'_x} \geq \fract \delta$,
%and $1^?_{\alpha, \beta}$ equals $1$ if $\alpha$ holds,
%$-1$ if $\beta$ holds, and $0$ otherwise.
%
Thus, constraints $\floor y - \floor {z_x} \leq k$ and $\fract y \leq \fract {z_x}$ in $\psi_\push$
can be replaced by combinations of constraints on $\floor{z'_y} - \floor {z'_x}$, $\fract {z'_y} \leq \fract{z'_x}$,
and $\alpha, \beta$.
}
(Notice that $\floor {z_0'}$ disappears in this process:
This is not a coincidence, since diagonal integer/modular/fractional constraints
are invariant under the elapse of an \emph{integer} amount of time.)
Thus by~\eqref{eq:floor:1} we obtain a constraint $\psi_\push'(\vec z'_{\vec x}, \vec z')$
logically equivalent to ${\psi_\push(\vec z'_{\vec x} - \vec z_0', \vec z' - \vec z_0')}$,
%
%
%Therefore, $\psi_\pop'(\vec x', \vec z'_{\vec x})$ is equivalent to
%$\exists \delta \geq 0, \vec z' \geq \vec 0 \st	\psi_\push'(\vec z'_{\vec x}, \vec z', \fract \delta) \wedge \psi_\pop(\vec x', \vec z')$,
and, by separating the fractional and integer constraints
(cf.~Remark~\ref{rem:constraints-as-qf-formulas}),
%\eqref{eq:psi:pop'}
	$\psi_\pop'(\vec x', \vec z'_{\vec x}) \equiv \exists \floor {\vec z'}, \fract {\vec z'} \st
	%&0 \leq \fract{\delta} < 1 \wedge \vec 0 \leq \floor {\vec z'} \wedge \vec 0 \leq \fract {\vec z'} < \vec 1\ \wedge \\
	\psi_\push'(\floor {\vec z'_{\vec x}}, \fract {\vec z'_{\vec x}}, \floor {\vec z'}, \fract {\vec z'}) \wedge
	\psi_\pop(\floor {\vec x'}, \fract {\vec x'}, \floor {\vec z'}, \fract {\vec z'})$.
%\end{align*}
%
By Corollary~\ref{cor:qe-clocks}, we can perform quantifier elimination
and we obtain a logically equivalent clock constraint of exponential size (in DNF)
%
%\begin{align}
%	\label{eq:xi}
$	\xi_{\psi_\push, \psi_\pop}(\floor{\vec x'}, \fract{\vec x'}, \floor{\vec z'_{\vec x}}, \fract{\vec z'_{\vec x}}),$
%\end{align}
%
where the subscript indicates that this formula depends on the pair $(\psi_\push, \psi_\pop)$ of push and pop constraints.
%
%We have thus obtained a new pop clock constraint $\xi_{\psi_\push, \psi_\pop}$.
The construction of $\P'$ consists in checking $\xi_{\psi_\push, \psi_\pop}$
in place of $\psi_\pop$, assuming that the push constraint was $\psi_\push$.
The latter is replaced by $\psicopy$.
Control states are the same in the two automata;
we can break down the $\xi_{\psi_\push, \psi_\pop}$ in DNF and record each conjunct in the stack,
yielding a new stack alphabet of exponential size.
%
% (SL camera-ready) The formal construction and its proof of correctness are presented in Sec.~\ref{app:construction:1}.

\begin{restatable}{lemma}{lemmaA}
	\label{lemma:A}
	Let $\set{\reach {} {\ell\arr}}_{\ell,\arr \in L}$,
	$\set{\reach {} {\ell\arr}'}_{\ell,\arr \in L}$
	be the reachability relations of $\P$, resp., $\P'$.
	Then, $\reach {} {\ell\arr} = \reach {} {\ell\arr}'$ for every $\ell, \arr \in L$,
	and $\P'$ has stack alphabet exponential in the size of $\P$.
\end{restatable}

\subsection*{\bf (C) The \TPDA is fractional}
Assume that the \TPDA $\P$ is both push-copy ({\bf A}) and pop-integer-free ({\bf B}).
%Recall that we assume that all modular constraints use the same modulus $M$.
%Assume additionally that all constants in integer constraints are strictly below $M$.
%
%Let $x_0$ be a global clock that is never reset.
%and let $z_0$ be a stack clock which has always fractional value $0$ at the time of push.
%
%Thanks to push-copy, diagonal pop modulo ${\floor {z_y} - \floor {z_x} \eqv m k}$
%and fractional $\fract {z_y} \leq \fract {z_x}$ constraints between stack clocks
%can be removed by checking ${\floor y - \floor x \eqv m k}$, resp., ${\fract y \leq \fract x}$, at the time of push.
%
We remove diagonal integer $\floor y - \floor x \sim k$
and modulo $\floor y - \floor x \eqv m k$ constraints on global clocks $x, y$
as in \TA \cite{AD94}.
%by introducing a new global clock for each such constraint.
%
In the rest of the section,
transition and stack constraints of $\P$ are of the form
\begin{align}
	\label{eq:transition:constraints:C}
	&\!\!\!\!\!\textrm{(trans.)}	&\floor x &\leq k, &\floor x &\eqv m k, &\fract x &= 0, &\fract x &\leq \fract y, \\%[1ex]
	&\!\!\!\!\!\textrm{(push)}	&\floor {z_x} &= \floor x, & & & & & \fract {z_x} &= \fract x, \\%[1ex]
	\label{eq:stack:constraints:C}
	&\!\!\!\!\!\textrm{(pop)}		& & &\floor y - \floor {z_x} &\eqv m k, &\fract {z_x} &= 0, &\fract y &\leq \fract {z_x}, \\
	\nonumber
	&		& & &\floor {z_y} - \floor {z_x} &\eqv m k, & & &\fract {z_y} &\leq \fract {z_x}.
\end{align}
%
%for $y$ a global clock or a corresponding stack clock $z_y$ in~\eqref{eq:stack:constraints:C}.
%

\vspace{-3mm}
\subparagraph{Unary abstraction.}
%
%Let $m \in \N$.
%
We replace the integer value of clocks by their \emph{unary abstraction}:
Valuations $\mu, \nu \in \Qgeq^X$ are \emph{$M$-unary equivalent},
written $\mu \approx_M \nu$,
if, for every clock $x \in X$,
$\floor {\mu(x)} \eqv M \floor{\nu(x)}$ and $\floor{\mu(x)} \leq M \iff \floor{\nu(x)} \leq M$.
Let $\Lambda_M$ be the (finite) set of $M$-unary equivalence classes of clock valuations.
For $\lambda \in \Lambda_M$ we abuse notation and write $\lambda(x)$ to indicate $\mu(x)$ for some $\mu\in\lambda$,
where the choice of representative $\mu$ does not matter.
We write $\lambda[Y \mapsto 0]$ for the equivalence class of
$\nu[Y \mapsto 0]$
%
%all clock valuations $\mu$ \st there exists a clock valuation $\nu \in \lambda$ \st $\mu(x) = 0$ for every $x \in Y$ and $\mu(x) = \nu(x)$ otherwise,
and we write $\lambda[x \mapsto x + 1]$ for the equivalence class of
$\nu[x \mapsto \nu(x)+1]$,
for some $\nu\in\lambda$ (whose choice is irrelevant).
%
%$\lambda$ extended \st for clock $x$ it yields the $M$-unary equivalence class of its successor $[\lambda(x)+1]_M$,
%and it is unchanged for the other clocks.
%in both cases the choice of $\nu\in\lambda$ is irrelevant.
%
Let $\varphi_\lambda(\vec x) \equiv \bigwedge_{x \in X } \floor x \eqv M \lambda(x) \wedge (\floor x < M \iff \lambda(x) < M)$
say that clocks belong to $\lambda$.
For $\varphi$ containing transition constraints of the form \eqref{eq:transition:constraints:C},
$\restrict \varphi \lambda$ is $\varphi$ where every integer $\floor x \leq k$ or modulo constraint $\floor x \eqv M k$
is uniquely resolved to be $\true$ or $\false$ by replacing every occurrence of $\floor x$ with $\lambda(x)$.
Similarly, for $\psi$ a pop constraint of the form \eqref{eq:stack:constraints:C},
$\restrict \psi {\lambda_\push, \lambda_\pop}$ is obtained by resolving modulo constraints $\floor y - \floor {z_x} \eqv M k$
and $\floor {z_y} - \floor {z_x} \eqv M k$
to be $\true$ or $\false$ by replacing every occurrence of $\floor y$ by its abstraction at the time of pop $\lambda_\pop (y)$,
and every occurrence of $\floor {z_x}$
by $\lambda_\push (x) + \Delta(\lambda_\push, \lambda_\pop)$,
i.e., the initial value of clock $x$ plus the total integer time elapsed until the pop, defined as
%
%\begin{align*}
	%
%	\floor y &\mapsto \lambda_\pop (y), \qquad \floor {z_x} \mapsto \lambda_\push (x) + \Delta(\lambda_\push, \lambda_\pop), \\[1ex]
	%
	$\Delta(\lambda_\push, \lambda_\pop) = \lambda_\pop (x_0) - \lambda_\push (x_0) - \condone {\fract {z_0} > \fract {x_0}}$,
%\end{align*}
%
i.e., we take the difference of $x_0$ (which is never reset) between push and pop,
possibly corrected by ``$-1$'' if the last time unit only partially elapsed;
the substitution for $\floor {z_y}$ is analogous.
%
%which can be detected by checking (at the time of pop) whether the fractional part of the stack clock $z_0$,
%which is initialised to $0$ at the time of push,
%is larger than the fractional part of $x_0$.
%
Fractional constraints are unchanged.

\vspace{-4mm}
\subparagraph{Sketch of the construction.}
Given a push-copy and pop-integer-free \TPDA $\P$,
we build a fractional \TPDA $\QQ$ over the extended alphabet
$\Sigma'= \Sigma \cup \setof {\tick x} {x \in X}$
as follows.
% (SL camera ready) (cf. Sec.~\ref{app:fractional} for the full construction).
%
We eliminate integer $\floor x \leq k$ and modulo constraints $\floor x \eqv M k$
by storing in the control the $M$-unary abstraction $\lambda$.
To reconstruct the reachability relation of $\P$,
we store the set of clocks $Y$ which will not be reset anymore in the future.
Thus, control locations $L'$ of $\QQ$ are of the form $\tuple {\ell, \lambda, Y}$.
In order to properly update the $M$-unary abstraction $\lambda$,
the automaton checks how much time elapses by looking at the fractional values of clocks.
When $\lambda$ is updated to $\lambda[x \mapsto x + 1]$,
a symbol $\tick x$ is optionally produced if $x \in Y$ was guessed not to be reset anymore in the future.
A test transition $\TPDArule \ell \varphi \arr$ is simulated by
$\TPDArule {\tuple {\ell, \lambda, Y}} {\restrict \varphi \lambda} {\tuple{\arr, \lambda, Y}}$.
A push-copy transition $\TPDArule \ell {\pushopgen{\alpha}{\psicopy}} \arr$
is simulated by %a push transition
$	\TPDArule
	{\tuple {\ell, \lambda, Y}}
	{\pushopgen{\tuple {\alpha, \lambda}} {\bigwedge_{x \in X} \fract {z_0} = 0 \wedge \fract {z_x} = \fract x}}
	{\tuple{\arr, \lambda, Y}}$
copying only the fractional parts and the unary class of global clocks.
A pop-integer-free transition $\TPDArule \ell {\popopgen{\alpha}{\psi}} \arr$
is simulated by % the fractional pop transition
$
	\TPDArule
	{\tuple {\ell, \lambda_\pop, Y}}
	{\popopgen{\tuple {\alpha, \lambda_\push}}{\restrict \psi {\lambda_\push,\lambda_\pop}}}
	{\tuple{\arr, \lambda_\pop, Y}}
$.
The reachability formula $\varphi_{\ell \arr}$ for $\P$
can be expressed by guessing the initial and final abstractions $\lambda, \mu$,
and the set of clocks $Y$ which is never reset in the run.
For clocks $x \in Y$, we must observe precisely $\floor {x'} - \floor x$ ticks $\tick x$,
and for the others, $\floor {x'}$, where $x$ is the initial and $x'$ the final value.
Let $g^Y_x = \floor {x'} - \floor x$ if $x \in Y$, and $\floor {x'}$ otherwise.
%
%By \emph{fractional} reachability we mean the reachability relation where the integer parts of initial and final values of clocks are ignored
%and only fractional parts are taken into account.
%(i.e., quantified existentially).
%Below, in a formula expressing reachability of a fractional \TPDA we omit the variables representing integer parts of clocks,
%as the integer parts of clocks are irrelevant.
\begin{restatable}{lemma}{lemC}
	\label{lem:C}
	Let $\set{\psi_{\ell' \arr'}(\fract{\vec x}, (\vec f, \vec g), \fract{\vec x'})}_{\ell', \arr' \in L'}$
	express the reachability relation of the fractional $\QQ$ %,
	where $\fract{\vec x}, \fract{\vec x'}$ are the fractional values of clocks
	(we ignore integer values),
	$\vec f$ is the Parikh image of the original input letters from $\Sigma$,
	and $\vec g$ of the new input letters $\tick x$'s.
	The reachability relation of $\P$ is expressed by
	%
	%\begin{align*}
	$\varphi_{\ell \arr}(\floor {\vec x}, \fract {\vec x}, \vec f, \floor {\vec x'}, \fract {\vec x'}) \equiv
			\bigvee_{\lambda, Y, \mu} \varphi_\lambda(\floor {\vec x}) \wedge
			\psi_{\tuple {\ell, \lambda, Y}\tuple {\arr, \mu, X}}(\fract {\vec x}, (\vec f, \vec g^Y), \fract {\vec x'})$.
%	\end{align*}
\end{restatable}
%
%\noindent
%From the complexity standpoint, the number of control states and stack symbols in the new automaton
%is exponential in the maximal constant $M$ and number of global clocks.
%
%The reachability formula $\varphi_{\ell \arr}$ is also of exponential size.

% !TEX root = main.tex

\section{From fractional \TPDA to register \PDA}
\label{sec:frac2reg}

The aim of this section is to prove the following result which, together with Theorem~\ref{thm:fractional:TPDA}, completes
the proof of our main result Theorem~\ref{thm:TPDA}.
\begin{restatable}{theorem}{thmFract2reg}
	\label{thm:fract}
	The fractional reachability relation of a fractional \TPDA $\P$
	is expressed by existential $\logic$ formulas, % $\psi_{\ell\arr}$
	computable in time
	exponential in the number of clocks and polynomial in the number of control locations and stack alphabet.
%	$\mu \reach {f} {\ell\arr} \nu$ in $\P$ if, and only if $\psi_{\ell\arr}(\fract \mu, f, \fract \nu)$.
\end{restatable}

\para{Cyclic atoms.}
We model fractional clock values by the \emph{cyclic atoms} structure
$
{(\unitint, K)}
$
with universe $\unitint = \Q \cap [0, 1)$,
where $K$ is the ternary cyclic order \eqref{eq:cyc}.
%
%${K(a, b, c) \equiv K(a, b, c) \wedge a \neq b \wedge a \neq c \wedge b \neq c.}$
%
%We could have taken and its non-strict variant $K_\leq(a,b,c) \equiv K(a,b,c) \vee a=b \vee b=c$
Since $K$ is invariant under cyclic shift,
it is convenient to think of elements of $\I$
as placed clockwise on a circle of unit perimeter; cf.~Fig.~\ref{fig:cyclic}(a).
An \emph{automorphism} is a bijection $\alpha$ that preserves and reflects $K$,
i.e., $K(a, b, c)$ iff $K(\alpha(a), \alpha(b), \alpha(c))$;
automorphisms are extended to tuples $\unitint^n$ point-wise.
%intuitively, $\alpha$ is a cyclic shift that can locally compress or stretch points.
%
% !TEX root = main.tex
\begin{wrapfigure}{r}{0.5\textwidth}
	\vspace{-20pt}
%\begin{figure}[h]
	\label{fig:cyclic}
	\begin{center}

		\begin{tikzpicture}[scale=0.8]
		\draw (-0.2,2) node {(a)};
		\draw (1,1) circle (1cm);
		\draw (1,2) node {\textbullet};
		\draw (1,1.75) node {0};
		\draw (0.4,1.8) node {\textbullet};
		\draw (0.4,1.6) node {$a$};
		\draw (1.95,1.3) node {\textbullet};
		\draw (1.75,1.3) node {$b$};
		\draw (1.4,0.1) node {\textbullet};
		\draw (1.4,0.35) node {$c$};
		\end{tikzpicture}
		\qquad\qquad
	\ignore{
		\begin{tikzpicture}[scale=0.8]
		\draw (-0.2,2) node {(b)};
		\draw (1,1) circle (1cm);
		\draw (1,2) node {\textbullet};
		\draw (1,1.75) node {$x_1$};
		\draw (0.4,1.8) node {\textbullet};
		\draw (0.5,1.55) node {$x_2$};
		\draw (1.95,0.7) node {\textbullet};
		\draw (1.35,0.95) node {$x_3{=}x_5$};
	%	\draw (1.4,0.1) node {\textbullet};
	%	\draw (1.4,0.3) node {$a_5$};
		\draw (0.6,0.1) node {\textbullet};
		\draw (0.6,0.35) node {$x_4$};
		\end{tikzpicture}
		\qquad\qquad
	}
		\begin{tikzpicture}[scale=0.8]
		\draw (-0.2,2) node {(b)};
		\draw (1,1) circle (1cm);
		\draw (1,2) node {\textbullet};
		\draw (1,1.75) node {0};
		\draw (0.3,1.7) node {\textbullet};
		\draw (0.35,1.5) node {$a$};
		\draw (2,1) node {\textbullet};
		\draw (1.75,1) node {$b$};
		\draw[very thick,blue] (0.3,1.7) arc (135:0:1) ;
		\draw[->,very thick,blue] (1.985,1.2) -- (2,1.05);
		\end{tikzpicture}
		\caption{
		(a) Relation $K$.
		%
		%(b) The characteristic formula of an orbit $o\in\orbits{}{\I^5}$:
		%$\varphi^\I_o(x_1,x_2,x_3,x_4,x_5) \equiv K(x_2,x_1,x_3) \wedge K(x_5,x_4,x_2) \wedge x_2 {\neq} x_1 {\neq} x_3{=}x_5 {\neq} x_4 {\neq} x_1$.
		(b) The cyclic difference $b\cycminus a$. % is interpreted as the length of the arc going clockwise from $a$ to $b$.
		}
		\label{fig:cyclic}
	\end{center}
  \vspace{-20pt}
%  \vspace{1pt}
\end{wrapfigure}%
%\end{figure}
%
%
%\noindent
Cyclic atoms are
%\footnote{Note that choosing another half-open unit interval $(0, 1]$ as the set of elements of $\I$, or even
%the open one $(0, 1)$, would yield an isomorphic structure.}
%
homogeneous~\cite{survey} %and hence admits quantifier elimination.
and thus $\unitint^n$ splits into exponentially many \emph{orbits} $\orbits{}{\I^n}$,
where $u, v\in\unitint^n$ are in the same orbit if some
automorphism maps $u$ to $v$.
An orbit is an equivalence class of indistinguishable tuples,
similarly as regions for clock valuations,
but in a different logical structure:
For instance $(0.2, 0.3, 0.7)$, $(0.7, 0.2, 0.3)$, and $(0.8, 0.2, 0.3)$ belong to the same orbit,
while $(0.2, 0.3, 0.3)$ belongs to a different orbit.
%
%
%Since an orbit $o\in\orbits{}{\I^n}$ is determined by the cyclic order of $n$ elements,
%and their equality type,
%the number of orbits is exponential in $n$. %; cf.~Fig.~\ref{fig:cyclic}(b).
%

\para{Register PDA.}
We extend classical pushdown automata with additional $\I$-valued \emph{registers},
both in the finite control (i.e., global registers) and in the stack.
Registers can be compared by quantifier-free formulas with equality and $K$,
called \emph{$K$-constraints}.
For simplicity, we assume that there are the same number of global and stack registers.
A \emph{register pushdown automaton} (\RPDA)
is a tuple $\QQ  = \tuple {\Sigma, \Gamma, L, X, Z, \Delta}$
where
$\Sigma$ is a finite input alphabet,
$\Gamma$ is a finite stack alphabet,
$L$ is a finite set of control locations,
$X$ is a finite set of \emph{global registers},
$Z$ is a finite set of \emph{stack registers}, and
the last item $\Delta$ is a set of transition rules
$\TPDArule \ell \op \arr$ with $\ell, \arr \in L$ control locations,
where $\op$ is either:
%
%\begin{inparaenum}
1) an input letter $a \in \Sigma_\varepsilon$,
%
%\item \emph{global}
2) a $2k$-ary $K$-constraint $\psi(\vec x, \vec x')$ relating pre- and post-values of global registers,
%
%\item \emph{push}
3) a push operation $\pushopgen{\alpha}{\psi(\vec x, \vec z)}$
with $\alpha \in \Gamma$ a stack symbol to be pushed on the stack
under the $2k$-ary $K$-constraint $\psi$ relating global $\vec x$ and stack $\vec z$ registers, or
%
%\item \emph{pop}
4) a pop operation $\popopgen{\alpha}{\psi(\vec x, \vec z)}$,
similarly as push.
%\end{inparaenum}
%
%\emph{input} $\op = a \in \Sigma_\varepsilon$ an input letter,
%%
%\emph{global transition} $\op = \psi(\vec x, \vec x')$ an $2k$-ary $K$-constraint relating pre- and post-values of global registers,
%%
%\emph{push} $\op = \pushopgen{\alpha}{\psi(\vec x, \vec z)}$ with $\alpha \in \Gamma$ a stack symbol to be pushed on the stack
%under the $2k$-ary $K$-constraint $\psi$ relating global $\vec x$ and stack $\vec z$ registers, and similarly,
%%
%\emph{pop} $\op = \popopgen{\alpha}{\psi(\vec x, \vec z)}$.
%
We consider \RPDA as symbolic representations of classical \PDA with infinite sets of control states
$\widetilde L = L \times \unitint^X$ and infinite stack alphabet $\widetilde \Gamma = \Gamma \times \unitint^Z$.
A \emph{configuration} is thus a tuple $\tuple{\ell, \mu, w} \in L\times \unitint^X \times \widetilde \Gamma^*$
where $\ell$ is a control location, $\mu$ is a valuation of the global registers, and $w$ is the current content of the stack.
Let $\tuple{\ell, \mu, u}, \tuple{\arr, \nu, v}$ be two configurations.
For every input symbol $a \in \Sigma_\varepsilon$ we have a transition
%
%\begin{gather*}
$ 	\tuple{\ell, \mu, u} \goesto a \tuple{\arr, \nu, v}$
%\end{gather*}
%
whenever there exists a rule $\TPDArule{\ell}{\op}{\arr} \in \Delta$ \st one of the following holds:
%
%\begin{itemize}
	%
%	\item[\bf (input)]
1) $\op = a \in \Sigma_\varepsilon$, $\mu = \nu$, $u = v$, or
	%
%	\item[\bf (global)]
2) $\op = \varphi$, $a = \varepsilon$, $(\mu, \nu) \models \varphi$, $u = v$, or
	%
%	\item[\bf (push)]
3) $\op = \pushopgen{\gamma}{\psi}$, $a = \varepsilon$, $\mu = \nu$,
	$v = u \cdot \tuple{\gamma, \mu_1}$
	if $\mu_1 \in \unitint^Z$ satisfies $(\mu, \mu_1) \models \psi$, or
	%
%	\item[\bf (pop)]
4) $\op = \popopgen{\gamma}{\psi}$, $a = \varepsilon$, $\mu = \nu$,
	$u = v \cdot \tuple{\gamma, \mu_1}$
	if $\mu_1 \in \unitint^Z$ satisfies $(\mu, \mu_1) \models \psi$.
%\end{itemize}

\para{Reachability relation.}
The reachability relations $\mu \reach w {\ell\arr} \nu$ and $\mu \reach f {\ell\arr} \nu$
are defined as for \TPDA by extending one-step transitions $\tuple{\ell, \mu, u} \goesto a \tuple{\arr, \nu, v}$
to words $w \in \Sigma^*$ and their Parikh images $f = \PI w \in \N^\Sigma$.
Thus, $\mu \reach f {\ell\arr} \nu$ is a subset of $\unitint^X \times \N^\Sigma \times \unitint^X$,
which is furthermore invariant under orbits.
In the following let $X'$ be a copy of global clocks.
An initial valuation $\mu$ belongs to $\I^X$,
a final valuation $\nu$ to $\I^{X'}$,
and the joint valuation $(\mu, \nu)$ belongs to $\I^{X \times X'}$.
The following two lemmas hold for \RPDA with homogeneous atoms;
cf.~\cite{ClementeLasota:CSL:2015}, or Sec.~9 in \cite{atombook}.
\begin{restatable}{lemma}{lemmaOrbitInv}
	\label{lem:orbit-inv}
	If $(\mu, \nu), (\mu', \nu')$ belong to the same orbit of $\unitint^{X \times X'}$,
	then $\mu \reach f {\ell\arr} \nu$ iff $\mu' \reach f {\ell\arr} \nu'$.
\end{restatable}
\minus
\begin{restatable}{lemma}{lemmaCfg}
	\label{lem:cfg}
	Given a \RPDA $\QQ$ %$ = \tuple {\Sigma, \Gamma, L, \Delta}$ of dimension $(k, l)$
	one can construct a context-free grammar $G$ of exponential size
	with nonterminals of the form $X_{\ell\arr o}$,
	for control locations $\ell, \arr$ and an orbit $o\in\orbits{}{\unitint^{X \times X'}}$, recognising the language
%
%	\begin{align} \label{eq:cfg}
$		L(X_{\ell\arr o}) = \setof
			{ \project \Sigma w \in \Sigma^* }
			{ \exists (\mu, \nu) \in o \st \mu \reach w {\ell\arr} \nu },$
%	\end{align}
	where $\project \Sigma w$ is $w$ without the $\varepsilon$'s.
	Consequently, \RPDA recognise context-free languages.
\end{restatable}
\minus
\begin{lemma}[Theorem 4 of \cite{VSS05}] \label{lem:ParIm}
	The Parikh image of $L(X_{\ell\arr o})$
	is expressed by an existential Presburger formula $\varphi^\Z_{\ell\arr o}$
	computable in time linear in the size of the grammar.
\end{lemma}
%
%\noindent From the last two lemmas we derive existential $\logicint$-formulas
%defining the Parikh image of $L(X_{\ell\arr o})$.
%Denote by $\varphi^\I_o$ the characteristic $K$-constraint of the orbit
%$o \in \orbits{}{{\I^{X \times X'}}}$ which specifies the cyclic order and the equality type of $2\card X$ elements.
%(cf.~Fig.~\ref{fig:cyclic}(b)).
%
%Relying on Lemma~\ref{lem:orbit-inv}, we deduce that  the reachability relation is definable in
%a two sorted logic: $\logicint$ on the integer sort, and $K$-constraints on the $\I$ sort.
%
\minus
\begin{restatable}{corollary}{corReach}
	\label{cor:reach}
	Let $\varphi^\I_o$ be the characteristic $K$-constraint of the orbit
	$o \in \orbits{}{{\I^{X \times X'}}}$.
	The reachability relation $\reach {} {\ell\arr}$ of an \RPDA $\QQ$ %of dimension $(k,l)$ if, and only if the following formula holds:
	is expressed by
	%
	%\begin{align} \label{eq:form}
		$\varphi_{\ell\arr}(\vec x, \vec f, \vec x') \ \equiv \ \bigvee_{o \in \orbits{}{{\I^{X \times X'}}}}
		\varphi^\Z_{\ell\arr o}(\vec f) \wedge \varphi^\I_o(\vec x, \vec x')$.
	%\end{align}
	%
	The size of $\varphi_{\ell\arr}$ is % polynomial in the size of $\QQ$ and the cardinality of $\orbits{}{2k}$.
	exponential in the size of $\QQ$.
\end{restatable}

\begin{proof}[Proof of Theorem~\ref{thm:fract}]
	Define \emph{cyclic sum} and \emph{difference} of $a, b\in \Q$ to be $a \cycplus b = \fract {a + b}$, resp., $a \cycminus b := \fract{a - b}$.
	For a set of clocks $X$,
	let $X_{x_0} = X \cup \set {x_0}$
	be its extension with an extra clock $x_0\notin X$ which is never reset,
	and let $\hat X_{x_0} = \setof {\hat x} {x \in X_{x_0}}$ be a corresponding set of registers.
	The special register $\hat x_0$ stores the (fractional part of the) current timestamp,
	and register $\hat x$ stores the (fractional part of the) timestamp of the last reset of $x$.
	In this way we can recover the fractional value of $x$
	as the cyclic difference $\fract x = \hat x_0 \cycminus \hat x$.
	Let (cf.~Fig.~\ref{fig:cyclic}(b))
	\begin{align}
		\label{eq:cycminus}
		\varphi_\cycminus(\vec x, \vec {\hat x}) \equiv \bigwedge_{x\in X} \fract x = \hat x_0 \cycminus \hat x.
	\end{align}
	Resetting clocks in $Y \subseteq X$ is simulated by
	$\varphi_{\resetop Y} \equiv
	\hat x_0' = \hat x_0 \wedge \bigwedge_{x \in Y} \hat x' = \hat x_0 \wedge \bigwedge_{x \in X\setminus Y} \hat x' = \hat x$
	and time elapse by
	$\varphi_\elapse \equiv \bigwedge_{x \in X} \hat x' = \hat x$.
	%
	%Thus registers corresponding to clocks in $Y$ are made equal $r_0$ ('now'), and other registers stay unchanged.
	The equality $\hat x_0' = \hat x_0$ in $\varphi_{\resetop Y}$
	says that time does not elapse,
	and the absence of constraints on $\hat x_0, \hat x_0'$ in $\varphi_\elapse$
	allows for an arbitrary elapse of time.
%	\begin{align} %\label{eq:cyc}
%	\fract {x_i} \  = \  r_0 \cycminus r_i, \qquad  \text{where} \quad a \cycminus b \  = \  \begin{cases}
%	a - b & \text{ if } a \geq b \\
%	a - b + 1 & \text { if } a < b.
%	\end{cases}
%	\end{align}
%Likewise, the stack registers $\{s_1, \ldots, s_l\}$ of $\QQ$ correspond to stack clocks of $\P$, and we have $z_i = r_0 \cycminus s_i$.
	%
	%As for now we have only used equality in $K$-constraints.
	%Now we will translate clock constraints into $K$-constraints,	using the cyclic order $K$.
	A clock constraint $\varphi$ is converted into a $K$-constraint $\hat \varphi$
	by replacing $\fract x = 0$ with $\hat x = \hat x_0$
	and $\fract x \leq \fract y$ by $K_\leq(\hat y, \hat x, \hat x_0)$,
	for $x, y \in X \cup Z$.
	%So prepared, we simulate every test, push or pop transition rule from $\Delta$ by replacing its clock constraint $\varphi$ by
	%the translated one $\widetilde \varphi$.
	%
	%\para{Construction of the $\C$-PDA.}
	%
	For a \TPDA $\P = \tuple {\Sigma,\Gamma, L, X, Z, \Delta}$,
	we define the following \RPDA
	$\QQ = \tuple {\Sigma, \Gamma, L, \hat X_{x_0}, \hat Z, \hat \Delta}$.
	%
 	%over the same input, stack alphabet, and control locations as $\P$.
	%
	%The transition rules $\Delta'$ are obtained from $\Delta$ as follows.
	%
	The input rules are preserved.
	A reset rule $\TPDArule \ell {\resetop Y} \arr \in \Delta$,
	is simulated by $\TPDArule \ell {\varphi_{\resetop Y}} \arr \in \hat \Delta$,
	a time elapse rule $\TPDArule \ell \elapse \arr \in \Delta$
	is simulated by $\TPDArule \ell {\varphi_\elapse} \arr \in \hat \Delta$,
	a push rule $\TPDArule \ell {\pushopgen \gamma \varphi} \arr \in \Delta$
	is simulated by $\TPDArule \ell {\pushopgen \gamma {\hat\varphi}} \arr \in \hat \Delta$,
	and similarly for pop rules.
	%
	%\para{Reconstruction of the reachability relation.}
	%
	By Corollary~\ref{cor:reach},
	let $\varphi_{\ell\arr}(\vec {\hat x}, \vec f, \vec {\hat x'})$ %as in~\eqref{eq:form}
	express the reachability relation of $\QQ$,
	and define $\xi_{o}^\I(\vec x, \vec x') \equiv
	 \exists \vec {\hat x}, \vec {\hat x'} \st
	 \varphi_{o}^\I(\vec {\hat x}, \vec {\hat x'}) \wedge
	 \varphi_\cycminus(\vec x, \vec {\hat x}) \wedge
	 \varphi_\cycminus(\vec x', \vec {\hat x'}).$
%	which \wlg we can assume of the form \slawek{no, you can't. You can only for a fixed orbit $o$}
%	$\varphi_{\ell\arr}^\Z(\vec f) \wedge \varphi_{\ell\arr}^\I(\vec {\hat x}, \vec {\hat x'})$,
%	by separating its $\Z$ and $\I$ parts.
	%
	The reachability relation of $\P$ is recovered as % (SL camera ready)  (the correctness is shown in Sec.~\ref{app:frac2reg:corr})
	\begin{align} \label{eq:corr}
		\begin{aligned}
		\psi_{\ell\arr}(\vec x, \vec f, \vec x')  \equiv \bigvee
		\{\varphi^\Z_{\ell\arr o}(\vec f) \wedge \xi^\I_o(\vec x, \vec x') | {o \in \orbits{}{{\I^{X \times X'}}}}\}.
		\end{aligned}
	\end{align}
	Intuitively, we guess the value for registers $\vec {\hat x}, \vec {\hat x'}$
	and we check that they correctly describe the fractional values of global clocks
	as prescribed by $\varphi_\cycminus$.
	We now remove the quantifiers from $\xi_{o}^\I$
	to uncover the structure of fractional value comparisons.
	Introduce a new variable $\delta = \hat x_0 \cycminus \hat x_0'$,
	and perform the following substitutions in $\varphi_{o}^\I$
	(\cf the definition of $\varphi_\cycminus$ in~\eqref{eq:cycminus}):
	%
	%\begin{align*}
		$\hat x		\mapsto \hat x_0 \cycminus \fract x$, % \qquad
		$\hat x'	\mapsto (\hat x_0 \cycminus \delta) \cycminus \fract {x'}$, and % \qquad
		$\hat x_0'	\mapsto \hat x_0 \cycminus \delta$.
	%\end{align*}
	%
	By writing $(\hat x_0 \cycminus \delta) \cycminus \fract {x'}$ as $\hat x_0 \cycminus (\delta \cycplus \fract {x'})$,
	we have only atomic constraints of the forms
	$K(\hat x_0 \cycminus u, \hat x_0 \cycminus v, \hat x_0 \cycminus t)$
	and $\hat x_0 \cycminus u = \hat x_0 \cycminus v$,
	where terms $u, v, t$ are of one of the forms
	$0$, $\fract x$, $\delta \cycplus \fract {x'}$, $\delta$.
	These constraints are equivalent, respectively, to $K(t, v, u)$ and $u = v$.
	By expanding the definition of $K$ (cf.~\eqref{eq:cyc}),
	we obtain only constraints of the form $u \precsim v$ with $\precsim \in \set {<, \leq}$.
	Since $\delta$ appears at most once on either side, it can either be eliminated if it appears on both $u, v$,
	or otherwise exactly one of $u, v$ is of the form $\delta$ or $\delta \oplus \fract {x'}$,
	and the other of the form $0$ or $\fract x$.
	By moving $\fract {x'}$ on the other side of the inequality in constraints containing $\delta \oplus \fract {x'}$,
	$\xi_{o}^\I$ is equivalent to
	$\bigwedge_i s_i \precsim t_i \ \wedge \ \exists 0 \leq \delta < 1 \st \bigwedge_j u_j \precsim \delta \wedge \bigwedge_k \delta \precsim v_k$,
	where the terms $s_i, t_i, u_j, v_k$'s are of the form $0$, $\fract x$, or $\fract x \cycminus \fract {y'}$.
	We can now eliminate the quantification on $\delta$ and get a constraint of the form $\bigwedge_h s_h \precsim t_h$.
	Finally, by expanding $b \cycminus a$ as
	$b - a + 1$ if $b < a$ and $b - a$ otherwise
	(since $a, b \in \unitint$)
	we have
	%
%	\begin{align*}
%		%
%		\psi_{\ell\arr}(\vec x, \vec f, \vec x') \ \equiv \
%			\varphi_{\ell\arr}^\Z(\vec f) \wedge \varphi_{\ell\arr}^\Q(\vec x, \vec x'),
%			\quad \textrm { with }
%		%
$
		\xi_{o}^\I(\vec x, \vec x') \ \equiv \
			\bigwedge_h s_h' \precsim t_h',
$
		%
%	\end{align*}
	%
	where the $s_h', t_h'$'s are of one of the forms: $0$, $\fract x$, $\fract x - \fract {y'}$, or $\fract x - \fract {y'} + 1$.
\end{proof}

\ignore{
\begin{remark}
	We did not use the full power of $\logicrat$ in defining the reachability relation of a fractional \TPDA $\P$:
	there is just one existential quantification (which can be notabene eliminated, cf.~\cite{FerranteRackoff:QE:Reals}) and
	additive expressions arising in~\eqref{eq:psi} are of very restricted forms:
	\slawek{check please if this is not contradictory with some known examples}
	\begin{align*}
	y_i & \sim y_j &
	y_i & \sim y'_j + \delta  &
	y_i & \sim \delta & y_i & \sim 0 \\
	y'_i & \sim y'_j & y_i +1 & \sim y'_j + \delta &
	y'_i + \delta & \sim 1 &  y'_i & \sim 0.
	\end{align*}
	The remarkable advantage of $K$-formulas over this fragment of $\logicrat$
	is that $K$ admits quantifier elimination, which is crucial for the construction of the
	grammar in Lemma~\ref{lem:cfg}.
\end{remark}
}

%\subparagraph*{Acknowledgements.}

\newpage
\bibliography{bib}

\appendix

% !TEX root = main.tex

\section{Appendix}

\subsection{Quantifier elimination}
\label{app:qe}

The following appeared as Lemma~\ref{lem:qe-union} in the main text.
\qeUnion*
\begin{proof}
	It suffices to consider a conjunctive formula of the form
	$\varphi \equiv \exists y \cdot \varphi_1 \wedge \varphi_2$
	where $\varphi_1$ is a quantifier-free $\str_1$-formula and $\varphi_2$ is a quantifier-free $\str_2$-formula.
	W.l.o.g.~suppose $y$ is quantified over $\str_1$.
	Since $y$ is a variable of the first sort, it does not appear free in $\varphi_2$,
	and thus $\varphi \equiv (\exists y \cdot \varphi_1) \wedge \varphi_2$.
	By assumption that $\str_1$ admits quantifier elimination,
	$\exists y \cdot \varphi_1$ is equivalent to a quantifier free formula  $\widetilde\varphi_1$,
	and thus the original formula $\varphi$ is equivalent to $\widetilde\varphi_1 \wedge \varphi_2$.
	It is easy to see that the complexities combine as claimed.
\end{proof}

Let $\sem\varphi$ be the set of valuations satisfying $\varphi$.

The following appeared as Lemma~\ref{lem:qe-clocks-int-rat} in the main text.
\qeIntRat*
We prove this by splitting it in two claims.

%\begin{restatable}{lemma}{qeInt}% [Quantifier elimination for $\L_\Z^*$ (cf. \cite{To:CSL:2009})]
\begin{lemma}
	\label{lem:qe-clocks-int-rat}
	The structure % $\AZsucczero$
	$\AZc$ admits effective elimination of quantifiers.
	The complexity is singly exponential for conjunctive formulas.
%\end{restatable}
\end{lemma}
\begin{proof}
	We assume that all modulo statements are over the same modulus $m$.
	%which is easily achieved by taking the least common multiplier of all moduli and by introducing some extra disjunctions.
	%
	%In order to show quantifier elimination,
	It suffices to consider a conjunctive formula of the form %(we write $y, x_i$ instead of $\floor y, \floor {x_i}$ for simplicity)
	\begin{align}
		\label{eq:integer:qe}
		\exists y \cdot \varphi \ \equiv \ \exists y \cdot \bigwedge_i x_i + \alpha_i \leq y \leq x_i + \beta_i \ \wedge \ y \eqv m x_i + \gamma_i,
	\end{align}
	where, for every $i$,
	$\alpha_i, \beta_i \in \Z \cup \set {-\infty, +\infty}$ with $\alpha_i \leq \beta_i$, $\gamma_i \in \set {0, \dots, m - 1}$,
	where for uniformity of notation we assume $x_0 = 0, \alpha_0 \geq 0$ in order to model non-diagonal constraints on $y$.
	If not all $\alpha_i$'s are equal to $-\infty$,
	then a satisfying $y$ will be of the form $x_j + \alpha_j + \delta$ with $\delta \in \set {0, \dots, m - 1}$
	where $j$ maximises $x_j + \alpha_j$.
	We claim that the following quantifier free formula $\widetilde\varphi$ is equivalent to \eqref{eq:integer:qe}:
	\begin{align}
		\label{eq:quantifier_free}
		\bigvee_{\delta \in \set {0, \dots, m - 1}} \bigvee_j
		\bigwedge_i x_i + \alpha_i \leq x_j + \alpha_j + \delta \leq x_i + \beta_i \ \wedge \ x_j + \alpha_j + \delta \eqv m x_i + \gamma_i.
	\end{align}
	For the complexity claim, $\widetilde\varphi$ is exponentially bigger than \eqref{eq:integer:qe} when constants are encoded in binary.
	%
	%Clearly, the formula above is a constraint.
	%
	For the inclusion $\sem{\widetilde\varphi} \subseteq \sem{\exists y \cdot \varphi}$,
	let $(a_1, \dots, a_n) \in \sem{\widetilde\varphi}$.
	There exist $\delta$ and $j$ as per \eqref{eq:quantifier_free},
	and thus taking $a_0 := a_j + \alpha_j + \delta$ yields
	$(a_0, a_1, \dots, a_n) \in \sem{\exists y \cdot \varphi}$.
	For the other inclusion,
	let $(a_0, a_1, \dots, a_n) \in \sem{\varphi}$.
	Let $j \neq 0$ be \st $a_j + \alpha_j$ is maximised,
	and define $\delta := a_0 - (a_j + \alpha_j) \mod m$.
	Clearly $\delta \geq 0$ since $a_0$ satisfies all the lower bounds $a_i + \alpha_i$.
	Since $a_0$ satisfies all the upper bounds $a_i + \beta_i$ and $a_j + \alpha_j + \delta \leq a_0$,
	upper bounds are also satisfied.
	Finally, since $a_0 \eqv m a_i + \gamma_i$ and $a_0 \eqv m a_j + \alpha_j + \delta$,
	we have that also the modular constraints $a_j + \alpha_j + \delta \eqv m a_i + \gamma_i$ are satisfied.
	Thus, we have $(a_1, \dots, a_n) \in \sem{\widetilde\varphi}$, as required.

	If all $\alpha_i$'s are equal to $-\infty$, then there are no lower bound constraints and only modulo constraints remain, hence
%\slawek{this case assumed $\N$; I've adapted to $\Z$}
%	the following quantifier free formula $\widetilde\varphi$ is equivalent to \eqref{eq:integer:qe}:
%	\[
%		\bigwedge_{i\neq 1} \ x_1 + \gamma_1 \eqv m x_i + \gamma_i.
%	\]
%	The complexity is polynomial in this case.
%
	and a satisfying $y$ 		(if it exists) can be taken in the interval $\set{0, \dots, m - 1}$, yielding
	\begin{align*}
		\bigvee_{\delta \in \set {0, \dots, m - 1}} \bigwedge_i \ \delta \leq x_i + \beta_i \ \wedge \
		\delta \eqv m x_i + \gamma_i.
	\end{align*}
	The same complexity holds.
	The formula above is shown equivalent to \eqref{eq:integer:qe} by a reasoning as in the previous paragraph.
\end{proof}

%\begin{lemma}[Quantifier elimination for $(\Z, \leq, +1)$]
%	\label{lem:integer:clock:qe}
%	Every formula $\varphi$ of the structure $(\Z, \leq, +1)$
%	is  effectively equivalent to a constraint $\widetilde\varphi$ of the same structure.
%\end{lemma}
%
%\begin{proof}
%	Notice that this result does not follow directly from Lemma~\ref{lem:integer:equivariant:qe},
%	since applying to a first-order formula $\varphi$ of $(\Z, \leq, +1)$
%	the quantifier elimination procedure for the more general structure $(\Z, \leq, \eqvs, +1)$
%	we obtain a quantifier-free formula of the latter structure,
%	and not of the former one,
%	i.e., modulo constraints could in principle be introduced.
%	%
%	However, by inspecting the proof of Lemma~\ref{lem:integer:equivariant:qe} we can see that this is not the case.
%\end{proof}

\begin{restatable}{lemma}{qeRat}%[Quantifier elimination for $\L_\Q^*$]
	\label{lem:qe-clocks-rat}
	The structure $\AIc$ admits effective elimination of quantifiers.
	The complexity is quadratic for conjunctive formulas.
\end{restatable}
%\qeRat
%
\begin{proof}
	It suffices to consider a conjunctive formula of the form $\varphi \equiv \exists y \cdot \bigwedge_k \varphi_k$
	where $\varphi_k$ are atomic rational formulas.
	If any $\varphi_k$ is the constraint $y = 0$,
	then we obtain $\widetilde\varphi$ by replacing $y$ with $0$ everywhere.
	Otherwise, $\varphi$ is of the form
	\begin{gather*}
		%\label{eq:integer:qe}
		\exists y \cdot \bigwedge_{i \in I} {x_i} \leq y \wedge \bigwedge_{j \in J} y \leq {x_j},
	\end{gather*}
	and we can eliminate $y$ by writing the equivalent constraint $\widetilde\varphi$
	\begin{gather*}
		\bigwedge_{i \in I} \bigwedge_{j \in J} {x_i} \leq {x_j}.
	\end{gather*}
	The size of $\widetilde\varphi$ is quadratic in the size of $\varphi$.
\end{proof}

\subsection{Characterisation of the reachability relation}

The following characterisation is used in the proof of Lemma~\ref{lemma:A}.

\begin{lemma}
	\label{lem:characterisation}
	%Let $\ell, \arr$ be control states of the \trPDA $\P$.
	%
	The relation $\reach {} {\ell\arr}$ is the least relation satisfying the following rules,
	for valuations $\mu, \nu, \mu', \nu' : \Q^X$
	and words $w, u, v \in \Sigma^*$:
	\begin{align*}
		&\textrm{\em (input)} &&
			\frac{}{\mu \reach a {\ell\arr} \mu}
				 \qquad &&\textrm{ if } \exists \TPDArule \ell a \arr \in \Delta \\[2ex]
		&\textrm{\em (test)} &&
		\frac{}{\mu \reach \varepsilon {\ell\arr} \mu}
			 \qquad &&\textrm{ if } \exists \TPDArule \ell \varphi \arr \in \Delta \st \mu \models \varphi \\[2ex]
		&\textrm{\em (reset)} &&
		\frac{}{\mu \reach \varepsilon {\ell\arr} \mu [Y \mapsto 0]}
			 \qquad &&\textrm{ if } \exists \TPDArule \ell {\resetop Y} \arr \in \Delta \\[2ex]
		&\textrm{\em (elapse)} &&
			\frac{}{\mu \reach \varepsilon {\ell\arr} \nu}
				 \qquad &&\textrm{ if } \exists \TPDArule \ell \elapse \arr \in \Delta, \delta > 0 \st \nu = \fract {\mu + \delta} \\[2ex]
		&\textrm{\em (push-pop)} &&
			\frac{\mu \reach w {\ell'\arr'} \nu}{\mu \reach w {\ell\arr} \nu}
				  \qquad &&\textrm{ if \eqref{eq:characterisation}} \\[2ex]
		&\textrm{\em (transitivity)} &&
			\frac{\mu \reach u {\ell\ell'} \mu' \ \mu' \reach v {\ell'\arr} \nu}{\mu \reach {uv} {\ell\arr} \nu}
	\end{align*}
	\begin{align}
		\label{eq:characterisation}
		\bigvee_{%\begin{array}{c}
		\substack{
			\TPDArule \ell {\pushopgen \gamma {\psi_\push}} {\ell'}, \\
			\TPDArule {\arr'} {\popopgen \gamma {\psi_\pop}} \arr \in \Delta
			%\end{array}
		}}
			\exists \mu_Z \in \Qgeq^Z,
			\exists \delta \in \Qgeq
			\st
			\left\{
			\begin{array}{l}
				(\mu, \mu_Z) \models \psi_\push(\vec x, \vec z) \ \wedge \\
				(\nu, \mu_Z + \delta) \models \psi_\pop(\vec x, \vec z).
		\end{array}\right.
	\end{align}
\end{lemma}

\subsection{Missing details for {\bf (A)} push-copy}
\label{app:push-copy}

\label{app:construction:1}

Let $\Xi$ be the set of all $\xi_{\psi_\push, \psi_\pop}$'s.
%We are now ready to provide the formal construction.
Let the original \TPDA be $\P = (\Sigma, \Gamma, L, X, Z, \Delta)$,
let $\Psi_\push$ be the set of all push constraints $\psi_\push$ of $\P$,
and let $\Psi_\pop$ be the set of all pop constraints $\psi_\pop$ of $\P$.
%
%Moreover, for every $\psi_\push \in \Psi_\push$ and $\psi_\pop \in \Psi_\pop$,
%let $\xi_{\psi_\push, \psi_\pop}$ be as in \eqref{eq:xi}.
%
We construct an equivalent \TPDA $\P' = (\Sigma, \Gamma', L, X, Z', \Delta')$ which only pushes on the stack copies of stack clocks.
Let $\Gamma' = \Gamma \times \Xi$,
%\setof {(\psi_\pop, \widetilde\xi_{\psi_\push, \psi_\pop})} {\psi_\pop \in \Psi_\pop, \widetilde\xi_{\psi_\push, \psi_\pop} \in\Xi}
$Z' = \setof {z_x} {x \in X}$,
and transitions in $\Delta'$ are determined as follows.

Every input, test, time elapse, and clock reset transitions in $\P$ generate identical transitions in $\P'$.
For every push transition $\TPDArule \ell {\pushopgen {\alpha} {\psi_\push}} \arr$ in $\P$,
we have a push transition in $\P'$ of the form
\begin{align*}
	\TPDArule \ell {\pushopgen{\tuple{\alpha, \xi_{\psi_\push, \psi_\pop}}}{\psicopy \wedge z_0 = 0}} \arr
\end{align*}
($z_0 = 0$ is compatible with push-copy by adding a new clock $x_0$ which is $0$ at the time of push and using $z_0 = x_0$;
we avoid this for simplicity)
for every guessed pop constraint $\psi_\pop \in \Psi_\pop$ of $\P$
and corresponding new pop constraint $\xi_{\psi_\push, \psi_\pop} \in \Xi$ %  (SL camera ready)  from \eqref{eq:xi},
and where $\psicopy$ is as in \eqref{eq:copy:push:constraint}.
%
%(Here, and in the pop transitions below, $Y$ could be taken to be $Y = \emptyset$ thanks to condition {(\bf A)}, but it is not necessary at this point.)
%
Finally, for every pop transition $\TPDArule \ell {\popopgen{\alpha}{\psi_\pop}} \arr$ in $\P$
and for every potential push constraint $\psi_\push \in \Psi_\push$,
%and for every $\xi_{\psi_\push, \psi_\pop} \in \Xi$,
we have a pop transition in $\P'$
\begin{align*}
	\TPDArule \ell {\popopgen{\tuple{\alpha, \xi_{\psi_\push, \psi_\pop}}}{\xi_{\psi_\push, \psi_\pop}}} \arr
\end{align*}
which checks that the pop constraint $\psi_\pop$ was indeed correctly guessed.

This translation preserves the reachability relation.
The following appeared as Lemma~\ref{lemma:A} in the main text.
\lemmaA*
\begin{proof}
	We prove
	$$\mu \reach w {\ell\arr} \nu \quad \iff \quad \mu \reachp w {\ell\arr} \nu $$
	by induction on the length of derivations,
	following the characterisation of Lemma~\ref{lem:characterisation}.
	Let $\mu \reach w {\ell\arr} \nu$ (the other direction is proved analogously).
	Since all transitions are the same except push and pop transitions,
	it suffices to prove it for matching pairs of push-pop transitions.
	By~\eqref{eq:characterisation},
	there exist transitions
	$\TPDArule \ell {\pushopgen \gamma {\psi_\push}} {\ell'}, \TPDArule {\arr'} {\popopgen \gamma {\psi_\pop}} \arr \in \Delta$,
	a stack clock valuation $\mu_Z \in \Qgeq^Z$,
	and a time elapse $\delta \in \Qgeq$ \st
	$(\mu, \mu_Z) \models \psi_\push(\vec x, \vec z)$,
	$(\nu, \mu_Z + \delta) \models \psi_\pop(\vec x', \vec z')$,
	and $\mu \reach w {\ell'\arr'} \nu$ in $\P$.
	By inductive hypothesis, $\mu \reach w {\ell'\arr'}' \nu$ in $\P'$.
	By construction, $\P'$ has matching transitions
	$\TPDArule \ell {\pushopgen{\tuple{\gamma, \xi_{\psi_\push, \psi_\pop}}}{\psicopy}} {\ell'}$
	and $\TPDArule {\arr'} {\popopgen{\tuple{\gamma, \xi_{\psi_\push, \psi_\pop}}}{\xi_{\psi_\push, \psi_\pop}}} \arr$.
	Clearly, $(\mu, \mu) \models \psicopy(\vec x, \vec z_{\vec x})$,
	where $z_x$ is the stack clock copying the value of clock $x$ at the time of push.
	Since stack clock $z_0$ was initially $0$,
	we have that its value at the end is exactly $\delta$.
	We show that
	$$(\vec x' : \nu, \vec z'_{\vec x} : \mu + \delta) \models \xi_{\psi_\push, \psi_\pop}(\vec x', \vec z'_{\vec x}),$$
	thus showing $\mu \reach w {\ell'\arr'}' \nu$ in $\P'$ by \eqref{eq:characterisation}.
	By its definition,
	$\xi_{\psi_\push, \psi_\pop}(\vec x', \vec z'_{\vec x})$
	is equivalent to $\psi_\pop'(\vec x', \vec z'_{\vec x})$ from \eqref{eq:psi:pop'}.
	Take $\mu_Z + \delta$ as the valuation for $\vec z'$,
	and we have
	$$(\vec x': \nu, \vec z' : \mu_Z + \delta, \vec z'_{\vec x} : \mu + \delta ) \models
		\psi_\push(\vec z'_{\vec x} - \vec \delta, \vec z' - \vec \delta) \wedge \psi_\pop(\vec x', \vec z')$$
	because $(\vec x': \nu, \vec z' : \mu_Z + \delta) \models \psi_\pop(\vec x', \vec z')$
	and $(\vec z_{\vec x} : \mu, \vec z : \mu_Z) \models \psi_\push(z_{\vec x}, \vec z)$.
\end{proof}

%\begin{remark}
%	The \TPDAmodel from \cite{ClementeLasota:LICS:2015}
%	only considered \emph{classical, non-diagonal} push constraints,
%	i.e., only boolean combinations of tests of the form $z \leq k$.
	%
%	Since classical constraints are expressible as
%	In this section, we showed that even additionally allowing diagonal classical push constraints $x - z \leq k$
%	causes the model to collapse to copy push constraints $x = z$.
%\end{remark}

%\subsection{Missing details for {\bf (B)} pop-integer-free}
\subsection{\bf (B) The \TPDA is pop-integer-free}
\label{app:pop:integer:free}

The aim of this section is to remove integer constraints from pop transitions.
Thanks to ({\bf A}), we assume that the \TPDA is push-copy.
Since diagonal integer constraints can simulate non-diagonal ones,
we can further assume that pop transitions do not contain non-diagonal integer constraints
(i.e., of the form $\floor z \leq k$),
and thus we only need to eliminate the diagonal ones.

Let $\P$ be a push-copy \TPDA. % $(\Sigma, \Gamma, L, X, Z, \Delta)$.
By Remark~\ref{rem:sugar}, we replace integer pop constraints of the form
${\floor x - \floor {z_y} \leq k}$, ${\floor {z_x} - \floor {z_y} \leq k}$
by \emph{classical} ${x - z_y \leq k}$, resp., $z_x - z_y \leq k$,
and fractional constraints.
This has the advantage that classical diagonal constraints are invariant under time elapse,
which will simplify the construction below.
Pop constraints of the form $z_y - z_x \sim k$ can easily be eliminated
since, thanks to push-copy, they can be checked at the time of push as the \emph{transition} constraint $y - x \sim k$.
Thus, we concentrate on pop constraints %$\popopgen \alpha {\psi_\pop}$
\begin{align}
	\label{eq:mixed:pop:constraint}
	\psi_\pop \ \equiv\
		\psi^{\textrm c}_1 \wedge \cdots \wedge \psi^{\textrm c}_m \wedge
		\psi^{\textrm {nc}}
\end{align}
where the $\psi^{\textrm c}_i$'s are classical diagonal constraints of the form ${y - z_x \sim k}$, with ${\sim \;\in\!\set {<, \leq, \geq, >}}$,
and $\psi^{\textrm {nc}}$ contains only non-classical (i.e., modular and fractional) constraints.
Let be $\C$ the set of all $\psi^{\textrm c}_i$'s.
Constraints $y - z_x \sim k$ are eliminated by introducing linearly many new global clocks (one for each atomic clock constraint)
satisfying suitable conditions at the time of push.
Thus, in the new automaton pop constraints are only of the form $\psi^{\textrm {nc}}$,
i.e., modulo and fractional, as required.
The construction is similar to \cite{ClementeLasota:LICS:2015}. %, and we present its details in Sec.~\ref{app:pop:integer:free};
Control states of the new automaton $\P'$ are of the form
$\tuple{\ell, T, \Phi^-, \Phi^+}$, where $T$ is a set of clocks
and $\Phi^-, \Phi^+$ are sets of atomic constraints.
Thus, from a complexity standpoint, the number of control locations of $\P'$ is exponential in the number of clocks and constraints,
and the size of the stack alphabet is exponential in the number of constraints.
%
%The following appeared as Lemma~\ref{lem:B} in the main text.

\begin{restatable}{lemma}{lemB}
	\label{lem:B}
	Let the reachability relation of $\P'$ be expressed by the formula $\varphi_{\ell'\arr'}$.
	The reachability relation of $\P$ is expressed by
	%
	%\begin{gather*}
	%	\reach{}{\ell,\arr} \ =\ \bigcup \setof
	%	 	{\reach {} {(\ell, T, \emptyset, \emptyset)(\arr,\emptyset, \Phi^-, \Phi^+)}}
	%		{T \subseteq X, \Phi^- \subseteq \C^-, \Phi^+ \subseteq \C^+}.
		$%\varphi_{\ell\arr} =
			\bigvee \setof
				{\varphi_{\tuple{\ell, T, \emptyset, \emptyset}\tuple{\arr,\emptyset, \Phi^-, \Phi^+}}}
				{T \subseteq X, \Phi^-, \Phi^+ \subseteq \C}$.
	%\end{gather*}
\end{restatable}

\begin{proof}
	Let $\P$ be a push-copy \TPDA $(\Sigma, \Gamma, L, X, Z, \Delta)$.
	Let $\C^-/\C^+$ be the set of all lower/upper bound classical pop constraints of the form
	$y - z_x \geq k, y - z_x > k$, or, resp., $y - z_x \leq k, y - z_x < k$,
	and let $\C = \C^- \cup \C^+$.
	We construct a \TPDA $\P' = (\Sigma, \Gamma', L', X', Z, \Delta')$
	with the same set of stack clocks as $\P$,
	and with global clocks being those of $\P$,
	plus a copy of each global clock for each lower/upper bound constraint:
	$X' := X \cup \setof {x_\psi} {\psi \in \C}$.
	A control location of $\P'$ is of the form
	$(\ell, T, \Phi^-, \Phi^+) \in L'$, where
	\begin{itemize}
		\item $\ell$ is a control location of $\P$,
		\item $T \subseteq X$ is a set of clocks of $\P$ which cannot be reset till the next push
		(this is used to guess and check last resets before a push), and
	    \item $\Phi^- \!\subseteq\! \C^-, \Phi^+ \!\subseteq\! \C^+$ are the currently \emph{active} lower/upper bound constraints.
		%
	    %\item $\Phi^-_{\textrm{old}} \subseteq \Phi^-$ and $\Phi^+_{\textrm{old}} \subseteq \Phi^+$
		%are the lower/upper bound constraints that were active at the time of the last push, and must thus later be restored.
	\end{itemize}
	The new stack alphabet $\Gamma'$ consists of tuples of the form $\tuple{\alpha, \Phi^-, \Phi^+}$
	with $\alpha \in \Gamma$ a stack symbol of $\P$ and $\Phi^-, \Phi^+$ as above.

	Let $\TPDArule \ell \op \arr$ be a transition in $\P$.
	If it is either an input $\op = a \in \Sigma_\varepsilon$, test $\op = \varphi$, or time elapse $\op = \elapse$ transition,
	then it generates corresponding transitions in $\P$ of the form
	${\TPDArule {(\ell, T, \Phi^-, \Phi^+)} \op {(\arr, T, \Phi^-, \Phi^+)}}$
	for every choice of $T, \Phi^-, \Phi^+$.
	A reset transition $\op = \resetop Y$ generates several reset transitions of the form
	$${\TPDArule {(\ell, T, \Phi^-, \Phi^+)}  {\resetop {Y \cup Y'}} {(\arr, T \cup U, \Phi^- \cup \Psi^-, \Phi^+ \cup \Psi^+)}}$$
	whenever
	\begin{enumerate}
		\item $Y \cap T = \emptyset$ (no forbidden clock is reset),
		\item $U \subseteq Y$ is a subset of reset clocks
		which are guessed to be reset for the last time till the next push,
		\item %$\Psi^- \subseteq \setof { (y - z_x \gtrsim k) \in \C^- \setminus \Phi^- } { x \in U }$
		$\Psi^- \subseteq \bigcup_{x \in U} \C^-_x \setminus \Phi^-$
		is a new set of lower bound constraints involving newly reset clocks in $U$, similarly
		\item %$\Psi^+ \subseteq \setof { (y - z_x \lesssim k) \in \C^+ \setminus \Phi^+ } { x \in U }$
		$\Psi^+ \subseteq \bigcup_{x \in U} \C^+_x \setminus \Phi^+$
		likewise for the upper bound constraints, and finally
		\item $Y' \subseteq \setof {x_\psi} {\psi \in \C}$
		contains all clocks relating to \emph{new} active lower bound constraints,
		and all clocks relating to (new or not) active upper bound constraints \wrt clocks $Y$ reset in this transition:
		\begin{align*}
		  Y'&= \setof {x_\varphi} { \varphi \in \Psi^- \textrm{ or } \varphi \in \Phi^+_Y \cup \Psi^+}, \textrm{ where } \\
		  \Phi^+_Y&= \setof {(y - z_x \lesssim k) \in \Phi^+ } { x \in Y }.
		\end{align*}
	\end{enumerate}
	A push transition $\op = \pushopgen {\alpha} {\psicopy}$
	(where $\psicopy$ is defined in \eqref{eq:copy:push:constraint}),
	generates a transition in $\P'$ of the form
	\begin{align*}
	 \TPDArule
		{(\ell, T, \Phi^-, \Phi^+)}
		{\pushopgen{\tuple{\alpha, \Phi^-, \Phi^+}}{\psicopy}}
	 	{(\arr, T', \Phi^-, \Phi^+)}
	\end{align*}
	only if $T = X$, i.e., all clocks were correctly guessed to be reset for the last time till this push,
	and for every set of clocks $T' \subseteq X$ which are guessed not to be reset till the \emph{next} push.
	Moreover, we push on the stack the current set of guessed constraints $\Phi^-, \Phi^+$.
	Finally, a pop transition $\op = \popop{\alpha \models \psi_\pop}$ of $\P$
	with $\psi_\pop$ as in \eqref{eq:mixed:pop:constraint},
	generates in $\P'$ a test followed by a pop transition of the form (omitting the intermediate state)
	\begin{align*}
		\TPDArule
			{(\ell, T, \Phi^-, \Phi^+)}
			{\ \widetilde \psi; \ \popopgen{\tuple{\alpha, \hat \Phi^-, \hat \Phi^+}}{\psi^{\textrm {nc}}}}
			{\ (\arr, T, \hat \Phi^-, \hat \Phi^+)}
	\end{align*}
	for every $T \subseteq X$, $\hat \Phi^- \subseteq \C^-$, $\hat \Phi^+ \subseteq \C^+$, whenever
	$\Phi^- \cup \Phi^+ = \set {\psi^{\textrm c}_1 \wedge \cdots \wedge \psi^{\textrm c}_m}$,
	i.e., the guess of upper and lower bounds was indeed correct,
	and where $\widetilde \psi$ is defined as
	$\widetilde \psi \ \equiv\ \bigwedge \setof
			{y - x_{\psi^{\textrm c}_i} \sim k}
			{\psi^{\textrm c}_i \in \Phi^- \cup \Phi^+, \psi^{\textrm c}_i \equiv y - z_x \sim k }.$
	We have removed pop integer constraints $\psi^{\textrm c}_i$'s by introducing classical constraints in $\widetilde \psi$,
	and the latter can be converted into integer and fractional constraints according to Remark~\ref{rem:sugar}.
	Notice that the stack non-classical constraint $\psi^{\textrm {nc}}$ is preserved from $\P$ to $\P'$.
	Thus, we obtain a pop-integer-free \TPDA, as required.

	%\subparagraph{Reconstruction of the reachability relation.}

	%The following lemma states that the reachability relation $\reach{}{\ell\arr}$ of $\P$
	%can be reconstructed in terms of the reachability relation of $\P'$.
	%
	%\lemB
	%
	%\begin{proof}[Proof sketch]
	%	The proof follows the same argument for stack classical constraints as in \cite{ClementeLasota:TPDA:Arxiv:2015},
	%	except that now non-classical stack constraints (not considered in \cite{ClementeLasota:TPDA:Arxiv:2015})
	%	are kept unchanged.
	%\end{proof}
	%

	The number of control locations of $\P'$ is
	$\card {L'} = \card L \cdot 2^{\card X} \cdot 2^{2 \cdot \card {\mathcal C}}$,
	%exponential in the number of clocks and constraints of $\P$,
	the number of stack symbols of $\P'$ is
	% wrong!
	%$\card {\Gamma'} = \card \Gamma \cdot \card {\mathcal C}^2$,
	% corrected
	$\card {\Gamma'} = \card \Gamma \cdot 2^{2 \cdot \card {\mathcal C}}$,
	and the number of clocks of $\P'$ is
	$\card {X'} = \card X + \card {\mathcal C}$.
	Thus, $\P'$ has number of control locations and stack symbols exponential in the size of $\P$,
	and number of clocks linear in the size of $\P$.

	The construction can be proved correct by the same argument for stack classical constraints as in \cite{ClementeLasota:TPDA:Arxiv:2015},
	except that now non-classical stack constraints (not considered in \cite{ClementeLasota:TPDA:Arxiv:2015})
	are kept unchanged.
\end{proof}

\subsection{Missing details for {\bf (C)} fractional}
\label{app:fractional}

%We will use the concept of orbits, similarly as in Section~\ref{sec:frac2reg},
Recall the structure of fractional values $\AIc = (\unitint, \leq, 0)$.
% is used instead of $\strC$.
%
An \emph{automorphism} of $\AIc$ is a bijection $\alpha$ \st $\alpha(0) = 0$
and $a \leq b$ iff $\alpha(a) \leq \alpha(b)$;
in other words, $0$ is fixed,
but otherwise distances can be stretched or compressed monotonically.
The set $\unitint^X$ of (fractional parts of) clock valuations splits
into finitely many \emph{orbits}, where $u, v\in\unitint^n$ are in the same orbit if some
automorphism of $\AIc$ maps $u$ to $v$.
Note that an orbit $o$ is determined by the order of elements, their equality type, and their equalities with $0$;
hence the number of orbits is exponential in $\card X$.
For an orbit $o$, let its \emph{characteristic formula} be the following quantifier-free $(\unitint, \leq, 0)$ formula
\begin{align} \label{eq:orbit-form}
	\varphi_o(\vec x) \equiv \bigwedge_{\tilde a_i = 0} x_i = 0 \wedge \bigwedge_{\tilde a_i \leq \tilde a_j} x_i \leq x_j,
\end{align}
where $(\tilde a_1, \dots, \tilde a_n)$ is any fixed element of $o$
(by the definition of orbit, $\varphi_o$ does not depend on the choice of representative).

%For an orbit $o$, by $\succ o$ we denote the successor orbit w.r.t.~time elapse,
%i.e., For a fractional orbit $o \subseteq (\Q \cap [0, 1))^k$,
%let $\succ o$ be the unique fractional orbit $o_1$
%\st there exists $\delta > 0$ \st $o_1 = \fract{o + \delta}$,
%and, for all $0 < \delta' \leq \delta$, $o_1 = \fract{o + \delta'}$.

Let $\P = (\Sigma, \Gamma, L, X, Z, \Delta)$ be a push-copy and pop-integer-free \TPDA.
%let $X = \{x_1, \ldots, x_k\}$,
%and assume \wlg that the largest constant appearing in any constraint is smaller than $M$,
%and also that all modulo constraints are over the same modulus $M$.
%
We build a fractional \TPDA $\P' = (\Sigma', \Gamma', L', X, Z, \Delta')$
where $\Sigma'$ equals $\Sigma$
extended with an extra symbol $\tick x \not\in\Sigma$ for every clock $x$ of $\P$,
$\Gamma' = \Gamma \times \Lambda_M$ extends $\Gamma$
by recording the $M$-unary equivalence class of clocks which are pushed on the stack,
and $L' = L \times \Lambda_M \times 2^X \cup L_\bullet$,
where %the second component $\lambda \in \Lambda_M$ abstracts the integer value of clocks, and
$Y_1 \in 2^X$ is the set of clocks which are \emph{not} allowed to be reset any more in the future,
and $L_\bullet$ contains some extra control locations used in the simulation.
Every transition $\TPDArule \ell \op \arr \in \Delta$ generates one or more transitions in $\Delta'$
according to $\op$.
If $\op = a \in \Sigma_\varepsilon$ is an input transition,
then $\Delta'$ contains a corresponding input transition
$\TPDArule {\tuple {\ell, \lambda, Y_1}} a {\tuple{\arr, \lambda, Y_1}}$,
for every choice of $\lambda, Y_1$.
If $\op = \varphi$ is a test transition,
then $\Delta'$ contains a corresponding test transition
$$\TPDArule {\tuple {\ell, \lambda, Y_1}} {\restrict \varphi \lambda} {\tuple{\arr, \lambda, Y_1}},$$
where $\restrict \varphi \lambda$ contains only fractional constraints.
If $\op = \resetop Y$ is a reset transition,
then $\Delta'$ contains a reset transition
$$\TPDArule {\tuple {\ell, \lambda, Y_1}} {\resetop Y} {\tuple{\arr, \lambda[Y \mapsto 0], Y_1 \cup Y_2}}$$
provided that $Y \subseteq X \setminus Y_1$ (no forbidden clocks are reset),
and where $Y_2 \subseteq Y$ are declared to be reset now for the last time.
If $\op = \elapse$ is a time elapse transition,
then we have the following 4 groups of transitions:
\begin{enumerate}
%\noindent
%\textbf{(1)}
\item
	First, we silently go to control location $\tuple {\ell, \lambda, Y_1, 1}$ to start the simulation:
	\begin{align*}
	\langle \tuple {\ell, \lambda, Y_1},
		\varepsilon, \tuple {\ell, \lambda, Y_1, 1}\rangle.
	\end{align*}

%\noindent
%\textbf{(2)}
\item
	We test that the current orbit of fractional values is $o$,
	we let time elapse,
	and then we test that the new orbit is $o'$.
	We can reconstruct the set of clocks $Y_{o, o'}$
	which have just overflown and for which we need to update their unary abstraction
	as $Y_{o, o'} = \setof {x \in X} {o(x) > 0 \textrm { and } o'(x) = 0}$.
	This yields the following sequence of transitions,
	where we omit the intermediate states for conciseness:
	\begin{align*}
%		\label{eq:quasi:fractional}
		\langle \tuple {\ell, \lambda, Y_1, 1},
			(\varphi_o; \elapse; \varphi_{o'}),
		\tuple {\ell, \lambda, Y_1, Y_{o, o'}, 2}\rangle.
	\end{align*}

%\noindent
%\textbf{(3)}
\item
	For each clock that needs to be updated in $Y_{o, o'}$, we increment its unary abstraction one by one,
	and we optionally emit a tick if this clock was guessed not to be reset anymore in the future:
	\begin{align*}
		\TPDArule {\tuple {\ell, \lambda, Y_1, Y_2, 2}} {\tick x^?} {\tuple {\ell, \lambda[x \mapsto x + 1], Y_1, Y_2 \setminus \set x, 2}},
	\end{align*}
	where $\tick x^?$ equals $\tick x$ if $x \in Y_2 \cap Y_1$,
	and $\varepsilon$ if $x \in Y_2 \setminus Y_1$.

%\noindent
%\textbf{(4)}
\item
	When the unary class of all overflown clocks has been updated,
	we either return to the beginning of the simulation (in order to simulate longer elapses of time),
	or we quit:
	\begin{align*}
		\TPDArule {\tuple {\ell, \lambda, Y_1, \emptyset, 2}} \varepsilon {\tuple {\ell, \lambda, Y_1, 1}},
		\quad \TPDArule {\tuple {\ell, \lambda, Y_1, \emptyset, 2}} \varepsilon {\tuple {\arr, \lambda, Y_1}}.
	\end{align*}
\end{enumerate}
%
%\noindent
If $\op = \pushopgen{\alpha}{\psicopy}$ is a push-copy transition,
then $\Delta'$ contains a push transition copying only the fractional parts and the unary class of global clocks:
$$
	\TPDArule
	{\tuple {\ell, \lambda, Y_1}}
	{\pushopgen{\tuple {\alpha, \lambda}} {\bigwedge_{x \in X} \fract {z_0} = 0 \wedge \fract {z_x} = \fract x}}
	{\tuple{\arr, \lambda, Y_1}}.
$$
If $\op = \popopgen{\alpha}{\psi}$ is a pop-integer-free transition,
then $\Delta'$ contains a fractional pop transition of the form
$$
	\TPDArule
	{\tuple {\ell, \lambda_\pop, Y_1}}
	{\popopgen{\tuple {\alpha, \lambda_\push}}{\restrict \psi {\lambda_\push,\lambda_\pop}}}
	{\tuple{\arr, \lambda_\pop, Y_1}}.
$$
We eliminated all occurrences of $\floor x$ both from transition and push/pop stack constraints.
Thus, transition and stack constraints of $\P'$ are only fractional.

\subparagraph{Reconstruction of the reachability relation.}

We reconstruct the reachability relation of $\P$ from that of $\P'$ as follows.
%
%(measuring the integer elapse of time for clocks which are guessed not to be reset anymore).
%
The reachability relation $\reach {} {\ell\arr}$ of $\P$ is expressed as the $\L_{\Z,\Q}$ formula%

\begin{align*}
	\varphi_{\ell \arr}(\floor {\vec x}, \fract {\vec x}, \vec f, \floor {\vec x'}, \fract {\vec x'}) &\ \equiv \
	\bigvee_{\lambda, Y, \mu} \exists \vec g \cdot \varphi_\lambda(\floor {\vec x}) \wedge \varphi_{\textrm{step}} \wedge \varphi_{\textrm{end}}, \textrm { where } \\[1ex]
	\varphi_{\textrm{step}} &\ \equiv\ \psi_{\tuple {\ell, \lambda, Y}\tuple {\arr, \mu, X}}(\fract {\vec x}, (\vec f, \vec g), \fract {\vec x'}) \\
	\varphi_{\textrm{end}} &\ \equiv\	\bigwedge_{x \in Y} 	\floor {x'} =  \floor x + g_x \wedge
										\bigwedge_{x \not\in Y} \floor {x'} = g_x .
\end{align*}
\begin{itemize}
	\item The formula $\varphi_\lambda$ ensures that the initial integer value of clocks
	has the same unary class as prescribed by $\lambda$.
	\item The formula $\varphi_{\textrm{step}}$ invokes the fractional reachability relation of $\P'$
	where $g_x$ counts the number of marks $\tick x$ since clock $x$ was last reset.
	\item The formula $\varphi_{\textrm{end}}$ uniquely determines the final integer values $\floor {x'}$ of all clocks of $\P$:
	For those clocks $x \not\in Y$ which are ever reset during the run,
	the final value of its integer part $\floor {x'}$ equals the integer time $g_x$ that elapsed since the last reset;
	for those clocks $x \in Y$ which are not reset during the run,
	$\floor {x'}$ equals their initial value plus the time elapsed since the beginning.
\end{itemize}
We can eliminate the existential quantification on $\vec g$ from the formula above
by noticing that $\varphi_{\textrm{end}}$ uniquely determines $\vec g$
as a function of $\floor {\vec x}, \floor {\vec x'}$ and $Y$,
thus obtaining the equivalent $\L_{\Z, \Q}$ formula in the following lemma.
The following appeared as Lemma~\ref{lem:C} in the main text.
\lemC*
\noindent
In the statement above, $\vec g^Y$ is defined as follows:
\begin{gather*}
	g^Y_x \equiv	\left\{\begin{array}{ll}
			\floor {x'} - \floor x	&\textrm{ if } x \in Y \\
			\floor {x'}				&\textrm{ otherwise}.
		\end{array}\right.
\end{gather*}
\ignore{
	\begin{proof}
		In terms of the reachability relation, we show the following stronger statement,
		for every control locations $\ell, \arr \in Q$:
		For every clock valuations $\mu, \nu \in \Qgeq^X$ and input word $w \in (\Sigma')^* = (\Sigma \cup \setof {\tick x} {x \in X})^*$
		we have
		\begin{gather*}
			\mu \reach {\project \Sigma w} {\ell\arr} \nu \quad \iff \quad
			\exists Y \subseteq X \st
			\left\{
			\begin{array}{l}
				\fract \mu \reach w {\tuple {\ell, \lambda(\mu), Y}\tuple {\arr, \lambda(\nu), X}} \fract \nu \quad
				\textrm{ and } \\[1ex]
				\floor{\nu(x)} = \left\{
				\begin{array}{ll}
					\floor {\mu(x)} + \PI w (\tick x) 	& \textrm { if } x \in Y \\
					\PI w (\tick x)						& \textrm { otherwise.}
				\end{array}
				\right.
			\end{array}
			\right.
		\end{gather*}
		where $\project \Sigma \cdot : (\Sigma')^* \to \Sigma^*$ projects away the ticks $\tick x$'s,
		$\lambda(\mu)$ is the unary class of valuation $\mu$ (and similarly for $\lambda(\nu)$),
		and $\PI w (\tick x)$ is the number of ticks $\tick x$ in $w$.
	\end{proof}
}

\subsection{Missing proofs from Sec.~\ref{sec:frac2reg}}

The following appeared as Lemma~\ref{lem:cfg} in the main text.
\lemmaCfg*
\begin{proof}
	This is a special case of the following general fact:
	An equivariant orbit-finite \PDA over homogeneous atoms can be transformed into an equivariant
	orbit-finite context-free grammar (see \cite{atombook,ClementeLasota:CSL:2015}).
	For concreteness, we provide the productions of the grammar.
	For $o\in\orbits{}{\I^{X \times X'}}$ we write $o_1$ (resp.~$o_2$) for the projections of $o$ on the first (resp.~last) $k$ coordinates.
	For every input transition $\TPDArule \ell a \arr$ and $o$ \st $o_1 = o_2$
	we have in the grammar a production
	\begin{align*}
		\textrm{(input)} & \qquad X_{\ell\arr o} \from a.
	\intertext{
	For every global transition rule $\TPDArule \ell \varphi \arr$ and $o$
	\st $o \models \varphi$ we have a production
	}
		\textrm{(global)} & \qquad X_{\ell\arr o} \from \varepsilon.
	\intertext{
	For an orbit $o\in\orbits{}{\I^{X_1 \times X_2 \times X_3}}$ and $i, j \in \set {1, 2, 3}$,
	denote by $o_{ij} \in\orbits{}{\I^{X_i \times X_j}}$ the projection of $o$ to ($k$-ary) components $i,j$.
	For every orbit $o\in\orbits{}{\I^{X \times X' \times X''}}$ we have a production
	}
		\textrm{(transitivity)} & \qquad X_{\ell\arr o_{13}} \from X_{\ell\ell'o_{12}} \cdot X_{\ell'\arr o_{23}}.
	\intertext{Finally, for every pair of transitions
	$\TPDArule \ell {\pushop{\gamma \models \varphi}} {\ell'}, \TPDArule {\arr'} {\popop{\gamma \models \psi}} \arr \in \Delta$
	and orbit $o\in\orbits{}{\I^{X \times X' \times Z}}$
	\st $o_{13} \models \varphi$ and $o_{23} \models \psi$, we have a production}
		\textrm{(push-pop)} & \qquad X_{\ell\arr o_{12}} \from X_{\ell'\arr'o_{12}}.\qedhere
	\end{align*}
\end{proof}

\subsubsection{Correctness of the construction}
\label{app:frac2reg:corr}

We argue that $\QQ$ and $\P$ faithfully simulate each other by providing a variant of strong bisimulation between their configurations.
A configuration $\tuple{\ell, \mu, u}$ of $\P$ \emph{is consistent with} a configuration
$\tuple{\arr, \nu, v}$ of $\QQ$, if
\begin{itemize}
	\item they have the same control locations $\ell = \arr$,
	\item every global clock $x$ and the corresponding register $\hat x$ satisfy $\fract {\mu(x)} = \nu(\hat x_0) \cycminus \nu(\hat x)$,
	\item $u = (\gamma_1, \mu_1) \cdots (\gamma_n, \mu_n)$, $v = (\gamma_1, \nu_1) \cdots (\gamma_n, \nu_n)$
	and, for every $1 \leq i \leq n$, stack clock $z$ and corresponding register $\hat z$,
	we have $\fract {\mu_i(z)} = \nu(\hat x_0) \cycminus \nu_i(\hat z)$.
\end{itemize}
The consistency is not one-to-one, for two reasons: on the side of $\P$ the integer parts of clocks are irrelevant and hence
can be arbitrary; and on the side of $\QQ$ the configuration is unique only up to cyclic shift.

A configuration $\tuple{\arr, \nu, v}$ (of $\P$ or $\QQ$) is an \emph{$a$-successor} of $\tuple{\ell, \mu, u}$
if $\tuple{\ell, \mu, u} \goesto a \tuple{\arr, \nu, v}$ (in $\P$ or $\QQ$, resp.);
in $\P$, additionally, if $a\in\Qgeq$, then we call $\tuple{\arr, \nu, v}$ an $\varepsilon$-successor of $\tuple{\ell, \mu, u}$.
By inspection of the construction of $\QQ$ we deduce:
% that consistency is a (strengthened variant of) strong bisimulation between $\P$ and $\QQ$:
%
\begin{claim} \label{claim:cons}
Every configuration of $\P$ (resp.~$\QQ$) is consistent with some configuration of $\QQ$ (resp.~$\P$).
Moreover, for every pair of consistent configurations of $\P$ and $\QQ$, respectively,
and $a\in \Sigma_\varepsilon$,
every $a$-successor of one of the configurations is consistent with exactly one $a$-successor of the other one.
\end{claim}
Thus, once a pair of consistent configurations is fixed, the $a$-successors in $\P$ and $\QQ$ are in a one-to-one correspondence.
For the correctness of~\eqref{eq:corr} in Sec.~\ref{sec:frac2reg} observe that
a configuration $\tuple {\ell, \mu, \varepsilon}$ of $\P$ and a configuration $\tuple{\ell, \nu, \varepsilon}$ of $\QQ$
are consistent if, and only if,
$
(\mu, \nu) \models \varphi_\cycminus.
$

\end{document}